\documentclass[numbers,webpdf,imaiai]{ima-authoring-template}%

\graphicspath{{Fig/}}

\theoremstyle{thmstyletwo}%
\newtheorem{theorem}{Theorem}
\newtheorem{proposition}[theorem]{Proposition}%
\newtheorem{remark}{Remark}%

\numberwithin{equation}{section}

\begin{document}

\DOI{DOI HERE}
\copyrightyear{2021}
\vol{00}
\pubyear{2021}
\access{Advance Access Publication Date: Day Month Year}
\appnotes{Paper}
\copyrightstatement{Published by Oxford University Press on behalf of the Institute of Mathematics and its Applications. All rights reserved.}
\firstpage{1}


\title[A Reduced Landau-de Gennes Study for Nematic Equilibria in Three-Dimensional Prisms]{A Reduced Landau-de Gennes Study for Nematic Equilibria in Three-Dimensional Prisms}

\author{Yucen Han
\address{\orgdiv{Department of Mathematics and Statistics}, \orgname{University of Strathclyde}, \orgaddress{\street{16 Richmond St}, 
\postcode{G1 1XQ},
\state{Glasgow}, \country{United Kingdom}}}}
\author{Baoming Shi
\address{\orgdiv{School of Mathematical Sciences}, \orgname{Peking University}, \orgaddress{\postcode{100871}, \state{Beijing}, \country{China}}}}
\author{Lei Zhang
\address{\orgdiv{Beijing International Center for Mathematical Research, Center for Quantitative Biology, Center for Machine Learning Research}, 
\orgname{Peking University},
\orgaddress{\postcode{100871}, \state{Beijing}, \country{China}}}}
\author{Apala Majumdar* \address{\orgdiv{Department of Mathematics and Statistics}, \orgname{University of Strathclyde}, \orgaddress{\street{16 Richmond St}, 
\postcode{G1 1XQ},
\state{Glasgow}, \country{United Kingdom}}}}

\authormark{Yucen Han et al.}

\corresp[*]{Corresponding author: \href{email:email-id.com}{apala.majumdar@strath.ac.uk}}

\received{Date}{0}{Year}
\revised{Date}{0}{Year}
\accepted{Date}{0}{Year}


\abstract{We model nematic liquid crystal configurations inside three-dimensional prisms, with a polygonal cross-section and Dirichlet boundary conditions on all prism surfaces. We work in a reduced Landau-de Gennes framework, and the Dirichlet conditions on the top and bottom surfaces are special in the sense, that they are critical points of the reduced Landau-de Gennes energy on the polygonal cross-section. The choice of the boundary conditions allows us to make a direct correspondence between the three-dimensional Landau-de Gennes critical points and pathways on the two-dimensional Landau-de Gennes solution landscape on the polygonal cross-section. We explore this concept by means of asymptotic analysis and numerical examples, with emphasis on a cuboid and a hexagonal prism, focusing on three-dimensional multistability tailored by two-dimensional solution landscapes.}

\def\Q{\mathbf{Q}}
\def\P{\mathbf{P}}
\def\I{\mathbf{I}}
\def\n{\mathbf{n}}
\def\r{\mathbf{r}}
\def\zhat{\hat{\mathbf{z}}}
\def\xhat{\hat{\mathbf{x}}}
\def\yhat{\hat{\mathbf{y}}}
\def\_v{\mathbf{v}}

\maketitle
\section{Introduction}
Nematic liquid crystals (NLCs) are classical examples of mesophases that combine the fluidity of liquids with the ordering of crystalline solids \cite{dg}. NLCs are anisotropic materials in the sense that the constituent rod-like or asymmetric molecules tend to align along some locally preferred directions, referred to as nematic directors. The directors are distinguished material directions, so that NLCs have direction-dependent physical, mechanical and optical properties \cite{dg,lagerwallreview}. The directionality of NLCs make them the working material of choice for a range of electro-optic devices e.g. display devices, sensors, thermometers, photonics and more recently, NLCs are also used for artificial intelligence, for example in micro-robotics and sensors for bacterial systems \cite{jiang2021using, yaorobots}.

Mathematics can play a crucial role for designer NLC-based materials technologies. One aspect is to accurately predict the observable NLC configurations in prototype settings, that mimic contemporary experiments and applications. Secondly, we want to design NLC configurations with desired properties or structural characteristics i.e. we want to propose mathematical algorithms for stabilising a priori prescribed NLC configurations. Our work in this paper is a forward step in the second direction. In a batch of previous papers \cite{han2020reduced, han2021solution}, we carefully study NLC equilibria on two-dimensional (2D) polygons subject to tangent boundary conditions, for which the nematic director is tangent to the polygon edges. We work in the powerful Landau-de Gennes (LdG) framework, which was one of the reasons for awarding Pierre de Gennes the Nobel Prize for physics in 1991 \cite{dg, newtonmottram, wang2021modeling}. In the LdG framework, the NLC state is described by the LdG $\mathbf{Q}$-tensor order parameter which has five degrees of freedom in three-dimensional (3D) settings. The degrees of freedom contain information about the nematic directors and the degree of nematic order about the directors. In 2D settings, we often work in the reduced LdG framework (rLdG), for which we can employ the reduced LdG order parameter with only two degrees of freedom - to account for the nematic director in the plane and to account for the degree of order about the planar director \cite{han2020reduced}; full details are given in the next section.

In \cite{han2020reduced}, we study the rLdG model on 2D polygons. We study how the rLdG equilibria (which are minimisers of the rLdG free energy and model the physically observable configurations) depend on the polygon edge length. For example, on a square domain, the unique rLdG energy minimiser is the Well Order Reconstruction Solution (WORS), with tangent boundary conditions on the square edges, for small edge lengths comparable to the nematic correlation length \cite{kralj2014order}. The WORS is distinguished by two defect lines along the two square diagonals, and the defect lines partition the square domain into four sub-domains such that the nematic director is constant in each sub-domain. As the edge length increases, the WORS loses stability but exists as a rLdG critical point for all edge lengths. For large square domains, the authors report two classes of rLdG equilibria - the stable diagonal (D) solutions for which the director is aligned along one of the square diagonals , and the rotated (R) solutions for which the director rotates by $\pi$ radians between a pair of parallel square edges. There are two D and four R solutions, and the D solutions have lower rLdG energy than the R states. The interested reader is referred to \cite{luo2012, tsakonas2007} for more details. In \cite{yin2020construction}, the authors compute non energy-minimising saddle points of the rLdG energy; they label the saddle points in terms of their index or the number of negative eigenvalues of the Hessian of the rLdG energy about the saddle point. The authors compute the index of the WORS as a function of the square edge length, being index-$0$ for small edge lengths and the index increases as the edge length increases. The authors also report other saddle points, e.g. the $BD$-state with a pair of line defects along a pair of opposite square edges, and the $T$-state with a line defect along one square diagonal. The unstable saddle points connect the stable $D$ and $R$ solutions i.e. we can find pathways between the $D$ and $R$ solutions, mediated by the high-index unstable saddle points e.g. WORS, $BD$ and $T$ saddle points. These pathways are of relevance whilst studying the switching mechanisms or non-equilibrium dynamics of these toy polygon systems.

We perform analogous studies for a 2D hexagon and pentagon in \cite{han2020reduced, han2021solution}. For small edge lengths (comparable to the nematic correlation length), these polygons support the unique $Ring$ solution, with a single central $+1$-defect consistent with the tangent boundary conditions. As the edge length increases, the $Ring$-solution loses stability and on a $K$-polygon with $K$ edges, there are at least $\frac{K(K-1)}{2}$ stable rLdG equilibria (local minimisers of the rLdG free energy) for large polygons. These large domain equilibria are distinguished by the locations of the so-called ``splay" vertices, such that the director has a splay-like profile near the vertex. The stable rLdG equilibria have two splay vertices, under some physically relevant assumptions and hence, we obtain $\frac{K(K-1)}{2}$ equilibria for the different choices of the splay vertices. On a hexagon, we obtain three distinct classes of rLdG equilibria - $Para$, $Meta$ and $Ortho$, and the $Para$ states have the lowest energy for which the splay vertices are the furthest. On a pentagon, there are two classes of rLdG equilibria - the $Meta$ and the $Ortho$, and we observe analogues of the unstable $BD$-state for all $K$-polygons with $K \geq 4$. There is no analogue of the $WORS$ for $K\neq 4$.

In this paper, we study the rLdG model or critical points of the rLdG energy on three-dimensional prisms, with a polygon cross-section and tangent boundary conditions on the lateral surfaces. We fix the boundary conditions or impose Dirichlet boundary conditions on the top and bottom surfaces, and these boundary conditions are rLdG critical points on the two-dimensional polygon cross-section, consistent with the tangent boundary conditions. Tangent boundary conditions and/or stable high resolution nematic textures on the top and bottom prism surfaces, could potentially be experimentally realised by rubbing techniques and/or photoalignment and photopatterning techniques \cite{chigrinov2013photoaligning}. The first question concerns the relevance of the rLdG model in a three-dimensional setting i.e. how can we constrain the nematic director (or the leading eigenvector of the LdG $\Q$-order parameter) to be in the cross-section plane or to be two-dimensional, in the prism interior, for a 3D setting? The boundary conditions only ensure planar nematic directors on the boundary surfaces and not necessarily in the interior of the prism. One potential scenario is that we study NLCs with negative dielectric anisotropy inside the three-dimensional prisms, and apply an electric field in the transverse direction or normal direction to the polygon cross-section. The negative dielectric anisotropy coerces the NLC director to be orthogonal to the applied electric field. The NLC director will then relax into the plane of the polygon and we conjecture that the director remains in the plane of the polygon, after the field is removed. The second question concerns the choice of the boundary conditions - can we realistically fix the boundary conditions on the top and bottom surfaces to be specified rLdG critical points on the polygon cross-section. This is unclear but it is possible that for NLC materials with negative dielectric anisotropy, the system will relax into stable rLdG equilibria on the top and bottom surfaces when the applied electric field is removed i.e. the boundary conditions would correspond to stable rLdG equilibria on the polygon cross-section subject to tangent boundary conditions on the polygon edges as studied in \cite{han2020reduced, han2021solution}. This would correspond to the $D$ and $R$ solutions on a cuboid, or the $Para$-solutions on a prism with a hexagonal cross-section etc.

Labelling the Dirichlet boundary conditions on the bottom (top) prism surfaces by $\P^b$ ($\P^t$), we investigate the following question - can a 3D rLdG critical point for which the nematic director is planar, but depends on all three spatial coordinates, be constructed from a pathway between $\P^b$ and $\P^t$ on the 2D solution landscape? In other words, can we use pathways on the 2D rLdG solution landscapes on polygons to construct 3D rLdG critical points on prisms, with a polygon cross-section. The answer is affirmative, but not every 2D pathway corresponds to a 3D rLdG critical point and equally, there are 3D rLdG critical points that cannot be constructed from 2D pathways on 2D solution landscapes. There are hidden, subtle compatibility conditions that determine the configuration and the index of the 3D rLdG critical point. For example, we choose $\P^b$ and $\P^t$ to be two $D$ solutions on a cuboid, and we observe the unstable $WORS$-texture in the cuboid interior, which would not be possible in 2D settings. We also work with examples for which $\P^b$ and $\P^t$ are higher energy or unstable rLdG critical points on the 2D prism cross-section, and in these cases, we observe multistability in certain geometrical regimes i.e. when the prism cross-sectional dimensions and the prism height are sufficiently large. Multistability refers to multiple stable 3D equilibria on prisms, all of which maybe relevant for experiments, and these multiple equilibria are distinguished by defect lines running across the prism interior (along which the nematic director cannot be defined). We propose that one could use optical tweezers to manipulate the defect lines and induce transitions between the multiple equilibria, akin to the experimental situations reported in \cite{vskarabot2014manipulation}. We hope that the examples in this paper can be informative for future studies of this challenging problem.

In Section~\ref{sec:theory}, we describe the theoretical framework in detail. In Section~\ref{sec:cuboid}, we focus on the cuboid and use a combination of asymptotic and numerical methods to study 3D rLdG critical points to show how multistability can be tailored by the square edge length and prism height. We use different combinations of $(\P^b, \P^t)$ to illustrate the effects of the boundary conditions on the solution landscapes. In Section~\ref{sec:prism}, we generalise these results to a hexagonal prism and conclude with some perspectives in Section~\ref{sec:conclusions}.

\section{Theoretical framework}
\label{sec:theory}
The Landau-de Gennes (LdG) theory is one of the most powerful continuum theories for nematic liquid crystals (NLCs) in the literature ~(Section 2.1, 2.3, and 3.1 in \cite{dg}). It describes the nematic state by the LdG $\Q$-tensor order parameter, which is a macroscopic measure of the material anisotropy or directionality. Mathematically speaking, the $\Q$-tensor is a symmetric traceless $3\times 3$ matrix, $\Q = \sum \limits_{i = 1}^3 \lambda_i \mathbf{e}_i\otimes \mathbf{e}_i$, where the eigenvectors, $\mathbf{e}_i$, describe the preferred material directions or preferred directions of averaged molecular alignment, and the corresponding eigenvalues, $\lambda_i$, measure the degree of orientational order about the corresponding $\mathbf{e}_i$. The nematic phase is said to be (i) isotropic if~$\Q=0$, (ii) uniaxial if $\Q$
has a pair of degenerate non-zero eigenvalues (and one distinguished eigendirection with the non-degenerate eigenvalue) and (iii) biaxial if~$\Q$ has three distinct eigenvalues~(Section 2.1.2 in \cite{dg}).

In the absence of surface energies, a particularly simple form of the LdG energy is given by
\begin{equation}
    I_{LdG}[\Q]:=\int_{\Omega} \frac{L}{2}|\nabla\Q|^2 + f_B\left(\Q\right) \mathrm{dV},\label{eq:3Denergy}
\end{equation}
where the elastic energy density and the bulk energy density are given by:
\begin{equation}
    |\nabla\Q|^2:=\sum_{i,j = 1}^{3} Q_{ij,x}^2 +  Q_{ij,y}^2 +  Q_{ij,z}^2, f_B\left(\Q\right):=\frac{A}{2}tr\Q^2-\frac{B}{3}tr\Q^3+\frac{C}{4}\left(tr\Q^2\right)^2,
    \label{eq:fB}
\end{equation}
$\Omega \subset \mathbb{R}^3$ is the three-dimensional domain, $tr$ is the notation for trace, the variable $A = \alpha(T-T^*)$ is a rescaled temperature; $\alpha, L, B, C$ are positive material-dependent constants and $T^*$ is the characteristic nematic supercooling temperature. We employ the one-constant approximation for the elastic energy density, for which all spatial deformations are equally energetically expensive, so that $| \nabla \Q |^2 = \sum_{i,j = 1}^{3}Q_{ij,x}^2 +  Q_{ij,y}^2 +  Q_{ij,z}^2$. The rescaled temperature $A$ has three characteristic values:(i) $A = 0$, below which the isotropic phase $\mathbf{Q} = 0$ loses stability, (ii) the nematic-isotropic transition temperature, $A = B^2/27C$, at which $f_B$ is minimized by the isotropic phase and a continuum of uniaxial states with $s = s_+ = B/3C$ and n arbitrary, and (iii) the nematic superheating temperature, $A = B^2/24C$ above which the isotropic state is the unique critical point of $f_B$. For a given low temperature $A<0$ (temperature $T<T^*$), the minima of the bulk potential, $f_B$, belong to the set $\mathcal{N}:=\left\{\Q\in \mathcal{M}^{3\times 3}: Q_{ij} = Q_{ji}, Q_{ii} = 0, \Q = s_+(\n\otimes\n - \I/3)\right\}$, where 
\begin{equation}\label{eq:s+}
s_+ = \frac{B+\sqrt{B^2 - 24AC}}{4C}
\end{equation} 
and $\n\in \mathcal{S}^2$ arbitrary. In other words, the vacuum manifold $\mathcal{N}$ is a continuum of uniaxial $\Q$-tensors with constant eigenvalues determined by $s_+$ in (\ref{eq:s+}).

The physically observable configurations are modelled by local or global energy minimisers in an appropriately defined admissible space. The non energy-minimising critical points of \eqref{eq:3Denergy} are equally important, since they connect the energy minimisers on the solution landscape, and often dictate the non-equilibrium dynamics and selection of the energy minimiser for multistable systems. To define the degree of instability, we introduce the Morse index. The Morse index of a saddle point of a energy functional is the number of negative eigenvalues of the Hessian of the energy functional about the critical point \cite{milnor1969morse}. Stable critical points have index-$0$ i.e. they have no unstable eigendirections, whereas unstable index-$k$ saddle points have $k$-unstable eigendirections in the solution landscape.

To this end, we take our 3D domain to be $V_K = E_K\times[-\lambda h,\lambda h]$, which is a prism of height $2\lambda h$ and a regular polygonal cross-section $E_K$, with edge length $\lambda$. The parameter, $h$, is the ratio of the height to the width of $V_K$. When $K = 4$, the square domain $E_4 = [-\lambda,\lambda]^2$ with four vertices at $w_1 = (\lambda, \lambda)$, $w_2 = (-\lambda, \lambda)$, $w_3 = (-\lambda, -\lambda)$ and $w_4 = (\lambda, -\lambda)$, otherwise $E_K$ is a $K$-regular polygon with $K$ edges, centered at the origin with vertices at $w_k = (\lambda cos(2\pi(k-1)/K), \lambda sin(2\pi(k-1)/K))$, $k = 1,...,K$.

We non-dimensionalize the system as, $(\bar{x},\bar{y},\bar{z}) =  \left(\frac{x}{\lambda}, \frac{y}{\lambda}, \frac{z}{\lambda h} \right)$,
\begin{equation}
\label{eq:non-dimensionalized}
    F_0[\Q] := \int_{\bar{V}_K} \left(\frac{1}{2}\left|\nabla_{\bar{x}\bar{y}}\Q\right|^2 + \frac{1}{2h^2}\left|\Q_{,\bar{z}}\right|^2 + \frac{\lambda^2}{L} f_B\left(\Q\right)\right) \mathrm{d\overline{V}}
\end{equation}
where $\overline{V}_K:=\overline{E}_K\times [-1,1]$, with unit polygonal cross-section in xy-plane, $\overline{E}_K$, and $\nabla_{\bar{x}\bar{y}}\mathbf{Q} = (\mathbf{Q}_{\bar{x}},\mathbf{Q}_{\bar{y}})^T$. In the following, the bar is omitted for convenience.

Following the work in \cite{han2020reduced,han2021solution}, we set $B= 0.64\times10^4 N/m^2$, and $C= 0.35\times10^4 N/m^2$ \cite{newtonmottram} and work at a fixed low temperature, $A = -B^2/(3C)$. In \cite{canevari2017order}, the authors show that for $A = -B^2/3C$, the LdG free energy admits a family of critical points, $\Q_c$, with a fixed eigenvector $\zhat$ and a constant eigenvalue $-B/3C$ associated with $\zhat$, and hence, $\Q_c$ has only two degrees of freedom. In other words, for this special temperature, $A = -B^2/3C$, the LdG free energy has a family of critical points on polygonal prisms, $V_K$, defined by
\begin{equation} \label{eq:Qc}
\Q_c = \P - \frac{B}{3C} (2\zhat \otimes \zhat - \xhat \otimes \xhat - \yhat \otimes \yhat )
\end{equation}
where $\P$ is a symmetric traceless $2 \times 2$ matrix (the entries in the third row and column are zero), and $\P$ is a critical point of the rLdG energy defined below:
\begin{equation}\label{p_energy}
    F[\P]: = \int_{V_K}\frac{1}{2}\left|\nabla_{xy} \P\right|^2+\frac{1}{2h^2}\left|\P_{,\bar{z}}\right|^2 + \frac{\bar{\lambda}^2}{2C}\left(-\frac{B^2}{4C}tr\P^2+\frac{C}{4}\left(tr\P^2\right)^2\right) \mathrm{dV},
\end{equation}
and $\bar{\lambda}^2 = \frac{2C\lambda^2}{L}$. The energy (\ref{p_energy}) is simply the LdG energy \eqref{eq:non-dimensionalized} of the specific branch of critical points in \eqref{eq:Qc}. We refer to the $\P$-eigenvector with the largest positive eigenvalue as the ``nematic director" in the plane.
We drop the bar over $\lambda$ for the rest of the manuscript.


This manuscript focuses on the relationship between LdG critical points on 3D polygonal prisms, $V_K$, and solutions landscapes for the rLdG model on regular polygons, $E_K$, and hence, we use the temperature, $A = -B^2/ 3C$, as employed in our previous 2D work in \cite{han2020reduced,han2021solution}, which also allows for direct comparisons between the results in 2D and 3D respectively. In fact, the critical points of the rLdG energy on $E_K$ are simply $z$-invariant critical points of \eqref{p_energy} on $V_K$. The authors have extensively studied solution landscapes of the rLdG model on regular polygons, $E_K$, in a batch of papers \cite{han2020reduced, han2021solution}, in terms of the reduced LdG tensors, $\P$-matrices in \eqref{eq:Qc}.
In this reduced description, there are two degrees of freedom to describe the nematic director in the plane of the polygon and the degree of order about this direction respectively. In \cite{han2020reduced, han2021solution}, the authors compute pathways between competing energy minimisers on polygons, and the pathway is mediated by saddle points or unstable critical points of the rLdG energy. It is interesting to investigate whether these 2D pathways can be used to construct critical points of the 3D LdG energy in \eqref{eq:3Denergy}, on 3D prisms with a polygonal cross-section i.e. if we can stack the different 2D critical points on a 2D pathway to construct a 3D critical point on a 3D domain and if there are algorithms for using the 2D critical points as building blocks for self-assembling 3D structures? In fact, not every 2D pathway can be used to construct a 3D critical point of \eqref{eq:3Denergy} and this raises interesting questions about the compatibility of 2D critical points for 3D studies.

The Dirichlet boundary conditions on the top and bottom surfaces of $V_K$ are taken to be
\begin{gather}\label{tb}
\P = \P^{b}(x, y)\ on\ z = -1;\ \P = \P^{t}(x, y)\ on\ z = 1;
\end{gather}
where $\P^{t}$ and $\P^{b}$ are solutions of
\begin{align}\label{eq:top_bottom}
        \Delta_{xy}P_{11} &= \lambda^2\left(P_{11}^2+P_{12}^2-\frac{B^2}{4C^2}\right)P_{11},\nonumber \\
        \Delta_{xy}P_{12}&= \lambda^2\left(P_{11}^2+P_{12}^2-\frac{B^2}{4C^2}\right)P_{12}.
\end{align}
on $E_K$, where $\Delta_{xy} = (\cdot)_{,xx} + (\cdot)_{,yy}$, i.e. $\P^t$ and $\P^b$ are critical points of the rLdG energy on the cross-section $E_K$, which could be identified with the end-points of a 2D pathway on the rLdG solution landscape on $E_K$. For example, for $V_4$, $\P^{t}$ and $\P^{b}$ could be $D$ or $T$ solutions as reported in our previous work \cite{yin2020construction}, and $Meta$ on pentagon ($V_5$), and $Para$ and $Tri$ states on hexagon ($V_6$) in \cite{han2020reduced} (see Fig. \ref{2D_solution}).
It is noticeable that the solutions with natural boundary condition on the top and bottom surfaces as studied in \cite{canevari2020well}, do not necessarily satisfy the Dirichlet boundary conditions above.

We impose a Dirichlet boundary condition, $\P_l$, on the lateral surfaces of $V_K$: 
\begin{equation} \label{BC}
 \P\left(x, \, y, \, z\right) = \P_l\left(x, \, y\right) \qquad \textrm{for }
 \left(x, \, y\right)\in\partial E_K, \ z\in[-1, \, 1]
\end{equation}
However, there is a necessary mismatch at the corners/vertices.
We define the distance between a point on the lateral surface $(w,z)$ and the vertical edges $(w_k,z)$, for any $z\in[-1,1]$ as
\begin{equation}
    dist\left(w\right) = min\{||w-w_k||_2,k = 1,...,K\},\ (w,z)\ on\ \partial E_K\times[-1,1].\nonumber
\end{equation}
On a cuboid, we define the tangential Dirichlet boundary condition $\P = \P_l$ on lateral surfaces, away from the vertical edges to be
\begin{equation}
    \begin{aligned}
    &P_{11l}\left(w\right) =  \begin{cases}
    -\frac{B}{2C},\ dist\left(w\right)>\epsilon, w\ on\ x = \pm 1\,\\
    \frac{B}{2C},\ dist\left(w\right)>\epsilon, w\ on\ y = \pm 1,\\
    \end{cases}
    &P_{12l}\left(w\right) = 0, w\ on\ \partial E_4,
    \end{aligned}
    \label{Pl}
\end{equation}
   where $B/C$ is the value of $s_+$ in \eqref{eq:s+} for $A = -B^2/(3C)$, $0<\epsilon \ll 1$ is the size of mismatch region. For any other prism $V_K$, the same principle applies for defining the tangential Dirichlet boundary condition on lateral surfaces, and we omit it here. This lateral boundary condition is compatible with any stacks of 2D solutions i.e. solutions of \eqref{eq:top_bottom}.

\begin{figure}
    \begin{center}
        \includegraphics[width=0.8\columnwidth]{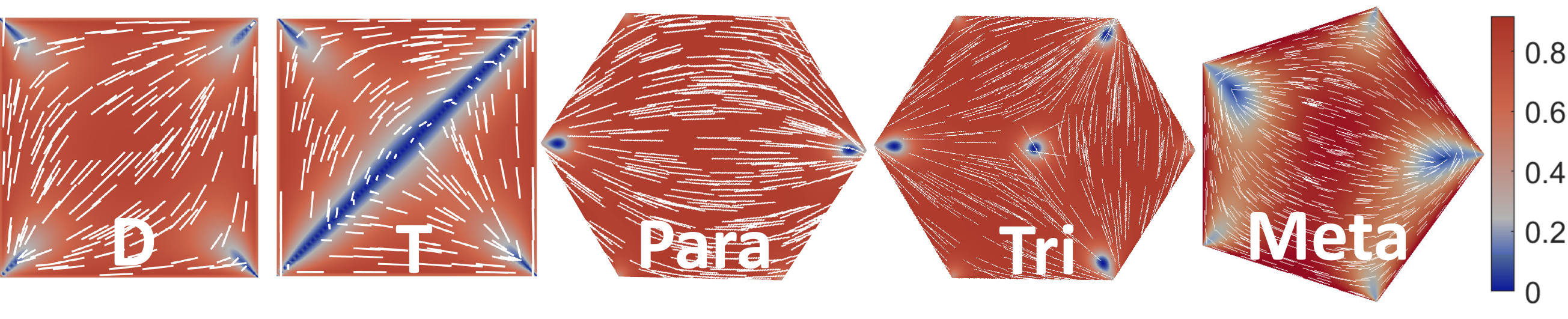}
        \caption{The profiles of 2D solutions of \eqref{eq:top_bottom} $D$ and $T$, on square with $\lambda^2=30$, $Meta$ on pentagon with $\lambda^2 = 30$, $Para$ and $Tri$ on hexagon with $\lambda^2=600$. The vector $(\cos(arctan(P_{12}/P_{11})/2),\sin(arctan(P_{12}/P_{11})/2))$ is the nematic director, and is plotted in terms of the white lines and the order parameter $\sqrt{P_{11}^2 + P_{12}^2}$ is represented by color from blue to red. }
        \label{2D_solution}
    \end{center}
\end{figure}

We take the admissible space to be
\begin{gather} \label{Ao}
    \mathcal{A}_0:=\{(P_{11},P_{12})\in W^{1,2}(V_K;\mathbb{R}^2) : \P=\P_l,\ \partial E_K\times[-1,1],\ \P = \P^{t},\ on\ z=1,\ \P = \P^{b},\ on\ z=-1\},
\end{gather} and the corresponding critical points, $\P(x,y,z)$ are solutions of the corresponding Euler--Lagrange equations:
\begin{align}\label{eq:EL}
        \Delta_{xy}P_{11} + \frac{1}{h^2}\Delta_z P_{11}&= \lambda^2\left(P_{11}^2+P_{12}^2-\frac{B^2}{4C^2}\right)P_{11}, \nonumber\\
        \Delta_{xy}P_{12} + \frac{1}{h^2}\Delta_z P_{12}&= \lambda^2\left(P_{11}^2+P_{12}^2-\frac{B^2}{4C^2}\right)P_{12}.
\end{align}
where $\Delta_{z} = (\cdot)_{,zz}$.
In what follows, we identify the defect set with the nodal set of solutions of \eqref{eq:EL} above. This is the set of no planar order i.e. for points in the nodal set of $\P$, the nematic director is not defined in the $(x,y)$-plane. This definition is widely employed for rLdG approaches as in \cite{han2020reduced}, \cite{han2021solution}. 

In the next proposition, we prove some basic existence and uniqueness results for critical points of \eqref{p_energy} in $\mathcal{A}_0$, before studying specific examples on cuboids and other generic prisms.
\begin{proposition}\label{unique}
For any $h$ and $\lambda$, there exist solutions of the Euler--Lagrange equations in \eqref{eq:EL}, in the admissible space $\mathcal{A}_0$ in \eqref{Ao}. For $h<h_0 = \frac{C}{2 B\lambda}$ or $\lambda<\lambda_0 = \frac{C}{2 Bh}$, the solution is unique, where $h$ and $\lambda$ are dimensionless/re-scaled measures of the prism height and cross-section dimension.
\end{proposition}
\begin{proof}
Our proof is analogous to Theorem 2.2 in \cite{bauman2012analysis}. Consider the LdG energy \eqref{p_energy} in terms of the two independent components, $P_{11}$ and $P_{12}$ of the reduced $\P$-tensor,
\begin{align}
    J[P_{11},P_{12}]:=&\int_{V_K} f_{el}(P_{11},P_{12})+f_b(P_{11},P_{12})\,\mathrm{dV},\label{funcq123}
\end{align}
where
\begin{gather}
    f_{el}(P_{11},P_{12}):=|\nabla_{xy} P_{11}|^2+|\nabla_{xy} P_{12}|^2+\frac{1}{h^2}|P_{11,z}|^2+\frac{1}{h^2}| P_{12,z}|^2,
\end{gather}
and
\begin{equation}\label{eq:fb}
    f_b(\P):=\frac{\lambda^2}{2C}\left(-\frac{B^2}{4C}|\mathbf{P}|^2 + \frac{C}{4}|\mathbf{P}|^4\right),
\end{equation}
are the elastic and thermotropic bulk energy densities, respectively. We prove the existence of minimizers of $J$ in the admissible class $\mathcal{A}_0$; the minimisers are necessarily solutions of (\ref{eq:EL}).
Since the boundary conditions are piece-wise of class $C^1$, the admissible space $\mathcal{A}_0$ is non-empty. $J$ is coercive in $\mathcal{A}_0$  since $|\nabla \P|^2$ is coercive. 
Finally, it suffices to note that $J$ is weakly lower semi-continuous on $W^{1,2}(V_K)$, which follows immediately from the fact that $f_{el}$ is quadratic and convex in $\nabla(P_{11},P_{12})$. 
Thus, the direct method in the calculus of variations yields the existence of a global minimizer of the functional $J$ in the space of finite-energy vectors, $(P_{11},P_{12})\in W^{1,2}(V_K; \mathbb{R}^2)$, satisfying the boundary conditions (\ref{BC}) and (\ref{tb})  (Theorem 2 in Section 8.2.2 of \cite{Evans49}). 
The semilinear elliptic system (\ref{eq:EL}) is simply the system of Euler-Lagrange equations associated with $J$, and the minimizers for $J$ are $C^{\infty}(V_K)\cap C^2(\overline{V_K})$ solutions of (\ref{eq:EL}). The minimising $\P$-tensor is an exact solution of the LdG Euler-Lagrange equations (\ref{eq:EL}).

We adapt the uniqueness criterion argument in Lemma 8.2 of \cite{lamy2014}. 
For any $B$, $C>0$ and $h>0$, if $(\P_{11},\P_{12})\in \mathcal{A}_0$ is a critical point of the rLdG energy \eqref{p_energy}, then $\P$ is bounded.
This is an immediate consequence of the maximum principle. We replace the operator $\nabla$ with $\mathcal{L}_h(.) = (\nabla_{xy}(.),\frac{1}{h}(.)_{,z})^T$ and following the calculations in the Lemma B.3. of \cite{lamy2014}, we have $|\P|^2 \leqslant \frac{B^2}{2C^2}$. We define the convex set $\mathcal{S}=\{(P_{11},P_{12})\in \mathcal{A}_0, |\P|^2\leqslant \frac{B^2}{2C^2}\}$. 

Then, we can prove that the functional $E$ is strictly convex on $\mathcal{S}$.
For any $\P, \bar{\P} \in \mathcal{S}$, we have
\begin{equation}
    \begin{aligned}
    &E(\frac{\P+\bar{\P}}{2})-\frac{1}{2}E(\P)-\frac{1}{2}E(\bar{\P}) \\
    &= \int_{V_K} -\frac{1}{8}|\nabla_{xy}(\P-\bar{\P})|^2  - \frac{1}{8h^2}|(\P-\bar{\P})_{,z}|^2 \mathrm{d}V + \int_{V_K} f_b(\frac{\bar{\P}+\P}{2})-\frac{1}{2}f_b(\P)-\frac{1}{2}f_b(\bar{\P})\mathrm{d}V.
    \end{aligned}
\label{E_convex}
\end{equation}
where $f_b(\P)$ is the bulk energy density in  \eqref{eq:fb}.
For any point $(\hat{x},\hat{y},\hat{z})\in V_K$, we have
\begin{equation}
     (P_{1i}-\bar{P}_{1i}) (\hat{x},\hat{y},\hat{z}) =\int_{-1}^{\hat{z}}(P_{1i}-\bar{P}_{1i})_{,z}(\hat{x},\hat{y},z)\mathrm{d}z,\ i = 1,2.
\end{equation}
Using the Cauchy-Schwarz inequality, we have
\begin{align}
     &(P_{1i}-\bar{P}_{1i})^2(\hat{x},\hat{y},\hat{z})  = \left(\int_{-1}^{\hat{z}}(P_{1i}-\bar{P}_{1i})_{,z}(\hat{x},\hat{y},z)dz\right)^2\\ &\leqslant |\hat{z}+1| \int_{-1}^{\hat{z}}(P_{1i}-\bar{P}_{1i})_{,z} ^2(\hat{x},\hat{y},z) \mathrm{d}z \leqslant 2\int_{-1}^{1} (P_{1i}-\bar{P}_{1i})_{,z} ^2(\hat{x},\hat{y},z) \mathrm{d}z.
\end{align}
Integrating both sides of the inequality on $V_K$, we have
\begin{align}\label{eq: poincare_pre}
    &\int_{V_K} (P_{1i}-\bar{P}_{1i}) ^2(\hat{x},\hat{y},\hat{z}) \mathrm{d}\hat{V}
    \leq  2\int_{-1}^1\int_{V_K} (P_{1i}-\bar{P}_{1i})_{,z}^2(x,y,z) \mathrm{d}V\mathrm{d}\hat{z} = 4\int_{V_K} (P_{1i}-\bar{P}_{1i})_{,z}^2(x,y,z) \mathrm{d}V,  
\end{align}
i.e. the Poincare inequality
\begin{equation}
    \Vert \P-\bar{\P}\Vert_{L^2(V_K)}^2\leqslant 4\Vert (\P-\bar{\P})_{,z}\Vert_{L^2(V_K)}^2,
    \label{eq: poincare}
\end{equation}
where we define the $L^2$-norm as $\Vert \P\Vert_{L^2(V_K)}= \left({\int_{V_K} |\P|^2}\text{d}V\right)^{\frac{1}{2}}$. The rationale of exchanging the order of integration in \eqref{eq: poincare_pre} follows from the density of $C_0^\infty(V_K)$ in $H_0^1(V_K)$, i.e., we can assume $P_{1i}-\bar{P}_{1i}\in C_0^\infty(V_K)$ \cite{majumdar2010landau}.

We compute an upper bound for the second integral in \eqref{E_convex}. 
\begin{align}
\left|f_b(\frac{\P+\bar{\P}}{2})-\frac{1}{2}f_b(\P)-\frac{1}{2}f_b(\bar{\P})\right| &\leq \frac{\lambda^2B^2}{8C^2}(-\left|\frac{\P+\bar{\P}}{2}\right|^2+\frac{1}{2}|\P|^2+\frac{1}{2}|\bar{\P}|^2)
+\frac{\lambda^2}{8}(-\left|\frac{\P+\bar{\P}}{2}\right|^4+\frac{1}{2}|\P|^4+\frac{1}{2}|\bar{\P}|^4)\nonumber\\
&\leq \frac{\lambda^2B^2}{32C}|\P - \bar{\P}|^2 +\frac{\lambda^2}{8}(-\left|\frac{\P+\bar{\P}}{2}\right|^4+\frac{1}{2}|\P|^4+\frac{1}{2}|\bar{\P}|^4)\nonumber
\end{align}
Since $|\P|^2|\bar{\P}|^2-\left<\P,\bar{\P}\right>^2\geqslant 0$, $|\P|^2,|\bar{\P}|^2\leq\frac{B^2}{2C^2}$, $\left<\P,\bar{\P}\right>\leq|\P||\bar{\P}|\leq \frac{B^2}{2C^2}$, we have
\begin{align}
&-\left|\frac{\P+\bar{\P}}{2}\right|^4+\frac{1}{2}|\P|^4+\frac{1}{2}|\bar{\P}|^4 \nonumber\\
&= \frac{7(|\P|^2+|\bar{\P}|^2)|\P-\bar{\P}|^2+10\left<\P,\bar{\P}\right>|\P-\bar{\P}|^2-16(|\P|^2|\bar{\P}|^2-\left<\P,\bar{\P}\right>^2)}{16}\nonumber\\
&\leq \frac{7(|\P|^2+|\bar{\P}|^2)|\P-\bar{\P}|^2+10\left<\P,\bar{\P}\right>|\P-\bar{\P}|^2}{16}\nonumber\\
&\leq \frac{3B^2}{4C^2}|\P-\bar{\P}|^2.\nonumber
\end{align}
Subsequently, 
\begin{equation}
\left|f_b(\frac{\P+\bar{\P}}{2})-\frac{1}{2}f_b(\P)-\frac{1}{2}f_b(\bar{\P})\right| \leqslant \frac{B^2\lambda^2}{8C^2}|\P-\bar{\P}|^2.
\end{equation}
By using the Poincare inequality in \eqref{eq: poincare}, we have
\begin{equation}
\int_{V_K}f_b(\frac{\P+\bar{\P}}{2})-\frac{1}{2}f_b(\P)-\frac{1}{2}f_b(\bar{\P})dV\leqslant\frac{B^2\lambda^2}{8C^2}\Vert \P-\bar{\P}\Vert_{L^2(V_K)}^2\leqslant \frac{B^2\lambda^2}{2C^2}\Vert (\P-\bar{\P})_{,z}\Vert_{L^2(V_K)}^2.
\end{equation}
Thus, for $h<h_0 = \frac{2C}{B\lambda}$ or $\lambda<\lambda_0 = \frac{2C}{Bh}$, the energy functional in \eqref{p_energy} is strictly convex on $\mathcal{S}$ and has a unique critical point,
since $\forall \P,\ \bar{\P} \in \mathcal{S}$, and $\P \neq \bar{\P}$,
\begin{equation}
    E(\frac{\P+\bar{\P}}{2})-\frac{1}{2}E(\P)-\frac{1}{2}E(\bar{\P})\leqslant - \frac{1}{8h^2}\Vert (\P-\bar{\P})_{,z}\Vert_{L^2(V_K)}^2+ \frac{B^2\lambda^2}{2C^2}\Vert (\P-\bar{\P})_{,z}\Vert_{L^2(V_K)}^2 < 0.
\end{equation}
\end{proof}

When $\lambda$ is small enough, the unique solution of the Euler--Lagrange equation \eqref{eq:top_bottom} on $E_K$ is a given $\P^*$ \cite{han2020reduced}. The boundary conditions on the top and bottom surfaces are solutions of \eqref{eq:top_bottom} on $E_K$, by choice. Hence, there is only one choice for  $\P^t$ and $\P^b$, defined by $\P^t= \P^b = \P^*$, for $\lambda$ sufficiently small. The z-invariant solution, $\P(x,y,z) = \P^*(x,y)$, is also a solution of the 3D Euler--Lagrange equations, \eqref{eq:EL}, on $V_K$. 
From Proposition $1$, the LdG energy has a unique critical point (or solution of \eqref{eq:EL}) on $V_K$, for $\lambda$ sufficiently small, and hence, this unique solution is the $z$-invariant 2D solution, defined by $\P(x,y,z) = \P^*(x,y)$.
We work with $\lambda$ large enough so that we can have $\P^t \neq \P^b$ and study mixed 3D critical points i.e. solutions of \eqref{eq:EL} on $V_K$ with conflicting boundary conditions on $z=\pm 1$. 

\section{The Cuboid, $V_4$}\label{sec:cuboid}

We consider two illustrative examples in this section, for two different choices of $\left(\P^b, \P^t \right)$, using a combination of analytic and numerical methods. For the first example, we take $(\P^b, \P^t) = \left(D1, D2 \right)$, for which the leading eigenvector of $\P$/ nematic director is almost aligned along one of the diagonals of the square cross-section, $E_4$. For $\lambda$ large enough, $D1$ and $D2$ are stable $z$-independent critical points of \eqref{p_energy} on $E_4$, subject to the boundary conditions, $\P_l$ on the square edges. For the second example, we take $(\P^b, \P^t) = \left(T1, T2 \right)$, where there are line defects with $\P^b, \P^t \approx 0$  on $y = x$ (see Fig. \ref{2D_solution}). The two-dimensional $D$ states are always index-$0$ or stable, and the $T$-states are always unstable, with Morse index-$3$ for $\lambda^2 = 30$, and necessarily have higher energy than the $D$ solutions. 
These two examples illustrate the dependence of 3D mixed critical points on the choices of $\P^t$ and $\P^b$, which is interesting since $\P^t$ and $\P^b$ could be experimentally tunable states or boundary effects.
\subsection{Choices of $\P^b$ and $\P^t$: $D1$ and $D2$}
\subsubsection{Small $h$}\label{small_h}
Let $\P^b = D1$ and $\P^t = D2$, so that the nematic director (leading eigenvector of $\P$) is aligned along $y=x$ for $D1$, and along $y = -x$ for $D2$.
We first note that the $\P$-tensors associated with $D1$ and $D2$ are solutions of \eqref{eq:top_bottom} and are related by \begin{equation}\label{12}(P_{11}^{D1},P_{12}^{D1}) = (P_{11}^{D2},-P_{12}^{D2}).
\end{equation} 
The $D1$ and $D2$-states have the reflection symmetry about the square diagonals, $x = y$ and $x = -y$, i.e. 
\begin{align}
&(P_{11}^{D_i}(y,x),P_{12}^{D_i}(y,x)) = (-P_{11}^{D_i}(x,y),P_{12}^{D_i}(x,y)),\label{yx}\\
&(P_{11}^{D_i}(-y,-x),P_{12}^{D_i}(-y,-x)) = (-P_{11}^{D_i}(x,y),P_{12}^{D_i}(x,y)),\ i = 1,2.\label{-y-x} 
\end{align}

For $h$ small enough, the solution of the Euler--Lagrange equation in \eqref{eq:EL} is unique, as in Proposition \ref{unique}. If $(P_{11}(x,y,z),P_{12}(x,y,z))$ is a solution of \eqref{eq:EL}, then so are $(P_{11}(x,y,-z),-P_{12}(x,y,-z))$ (from \eqref{12}), $(-P_{11}(y,x,z),P_{12}(y,x,z))$ (from \eqref{yx}), and $(-P_{11}(-y,-x,z),P_{12}(-y,-x,z))$ (from \eqref{-y-x}). 
Subsequently, on the middle cross-section of $V_4$, $(x,y,0)$ for $(x,y)\in E_4$, since $(P_{11}(x,y,z),P_{12}(x,y,z))=(P_{11}(x,y,-z),-P_{12}(x,y,-z))$, we have $P_{12}(x,y,0) = -P_{12}(x,y,0) \equiv 0$, $(x,y)\in E_4$. Since $P_{11}(x,y,z)= -P_{11}(y,x,z)$, we have $P_{11}(x,x,z) = -P_{11}(x,x,z) = 0$.  Since $P_{11}(x,y,z)= -P_{11}(-y,-x,z)$, we have $P_{11}(x,-x,z) = -P_{11}(x,-x,z) = 0$ for any $z$. Hence, for $h$ small enough, we have $\P(x,x,0) = \P(x, -x, 0) = 0$, with two line defects along the square cross-section on $z=0$. This is strongly reminiscent of the 2D solution of \eqref{eq:top_bottom}, known as the $WORS$ (Well Order Reconstruction Solution) \cite{kraljmajumdar2014} and henceforth, we refer to this mixed critical point as $D1 - WORS-D2$ in the rest of the paper. For $h$ small enough, this is the unique and hence, globally stable critical point of \eqref{p_energy}.

For more general cases, in the $h\to 0$ limit, we  can take a regular  perturbation expansion of $P_{11}$ and $P_{12}$ in powers of $h$ as shown below:
\begin{align}
P_{11}(x,y,z) &= P_{11}^0(x,y,z) + hf^0(x,y,z) + \mathcal{O}(h^2)\label{p11_expansion}\\
P_{12}(x,y,z) &= P_{12}^0(x,y,z) + hg^0(x,y,z) + \mathcal{O}(h^2)\label{p12_expansion}
\end{align}
for some functions $f^0, g^0$ which vanish on the boundary.
Substituting \eqref{p11_expansion} and \eqref{p12_expansion} into the Euler--Lagrange equations \eqref{eq:EL}, and multiplying the equations by $h^2$, we obtain
\begin{align}\label{eq:ELhsmall}
        h^2\Delta_{xy} P_{11} + \Delta_z P_{11}&= h^2\lambda^2\left(P_{11}^2+P_{12}^2-\frac{B^2}{4C^2}\right)P_{11},\\
        h^2\Delta_{xy} P_{12} + \Delta_z P_{12}&= h^2\lambda^2\left(P_{11}^2+P_{12}^2-\frac{B^2}{4C^2}\right)P_{12}.
\end{align}
The leading partial differential equations for $P_{11}^0$, $P_{12}^0$ are given by:
\begin{gather}\label{eq:h_0}
\Delta_z P_{11}^0 = 0,\ \Delta_z P_{12}^0.
\end{gather}
which admit the unique solution:  
\begin{gather}
P_{11}^0 = \frac{1-z}{2}P_{11}^{b} + \frac{1+z}{2}P_{11}^{t},\ P_{12}^0 = \frac{1-z}{2}P_{12}^{b} + \frac{1+z}{2}P_{12}^{t}.
\end{gather}
These expressions hold for all $V_K$ and $\P^b$, $\P^t$.
Up to $\mathcal{O}(h)$, the governing partial differential equations for $f$, $g$ are given by:
\begin{gather}
\Delta_z f= 0,\ \Delta_z g= 0,
\end{gather}
i.e. $f=g\equiv 0$, which means the first order corrections are zero.
The small difference between the limiting solution, $(P_{11}^0,P_{12}^0)$, and the unique solution, $(P_{11},P_{12})$, of the Euler--Lagrange equation for $h = 0.1$, in Fig. \ref{h_0}, indicates that the limiting solution is a good approximation of the unique solution.
\begin{figure}
    \begin{center}
        \includegraphics[width=0.5\columnwidth]{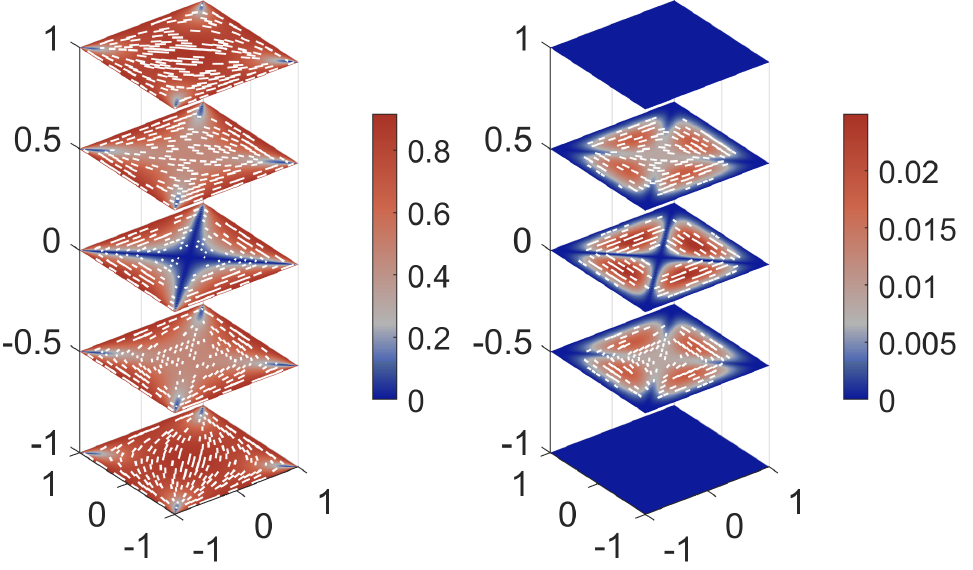}
        \caption{Left: the profiles for $D1-WORS-D2$, the solution of the Euler-Lagrange equation \eqref{eq:h_0}, $(P_{11},P_{12})$ with $h = 0.1$ and $\lambda^2 = 30$. The color represents order parameter $\sqrt{P_{11}^2+ P_{12}^2}$ and the white lines represent the nematic director $(cos(arctan(P_{12}/P_{11})/2),sin(arctan(P_{12}/P_{11})/2))$. Right: the profiles of the difference between the limiting solution as $h\to 0$, $(P_{11}^0,P_{12}^0)$, the solution of \eqref{eq:h_0} and the numerical solution on the left. 
        Let $(d_{11},d_{12}) = (P_{11},P_{12}) - (P_{11}^0,P_{12}^0)$. The white lines represent the vector field, $(\cos(arctan(d_{12}/d_{11})/2),\sin(arctan(d_{12}/d_{11})/2))$ and the color bar denotes the quantity, $\sqrt{d_{11}^2 + d_{12}^2}$. The white lines are drawn when $\sqrt{d_{11}^2 + d_{12}^2}\geq 4e-3$.}
        \label{h_0}
    \end{center}
\end{figure}

\subsubsection{The existence of $D1-WORS-D2$ for all $h$}\label{all_h}
\begin{proposition} \label{prop:forever_critical}
    Let $\lambda$ be large enough so that $D1$ and $D2$ are solutions of \eqref{eq:top_bottom}, and stable $z$-independent critical points of \eqref{p_energy} on a square domain, $E_4$, subject to the boundary conditions, $\P = \P_l$ on the square edges. With $\P^{t} = \P^{D2}$ and $\P^{b} = \P^{D1}$, $D1-WORS-D2$ is a critical point, $(P_{11}^s, P_{12}^s)$, of the energy functional \eqref{p_energy} on the cuboid, $V_4$, in the admissible space 
    $\mathcal{A}_0$ in \eqref{Ao}, for all $h>0$.
\end{proposition}
\begin{proof} 
We follow the approach in \cite{canevari2017order}. Consider a quadrant of the square domain, denoted by $\Omega_q$:
\begin{gather}
    \Omega_q:=\{(x,y)\in E_4 : -x<y<x,\ 0<x<1\}.
\end{gather}
The following boundary conditions on $\Omega_q\times[-1,0]$ are consistent with the boundary conditions (\ref{BC}) and (\ref{tb}), on the whole of $V_4$: 
\begin{gather}\label{ANquadrantBCs}
    \begin{cases}
    \P=\P_l,\ (x,y)\in\partial\Omega_q\cap\partial E_4,\ z\in[-1,0]; \\
    \P = \P^{D1},\ (x,y)\in\Omega_q,\ z = -1;\\
    P_{11}=\partial_\nu P_{12} = 0,\ (x,y)\in\{\partial\Omega_q\cap\{y = \pm x\}\}\times[-1,0];\\
    P_{12} = P_{11,z} = 0,\ (x,y)\in\Omega_q,\ z = 0,
    \end{cases}
\end{gather}
where $\partial_\nu$ represents the outward normal derivative in $xy$-plane. The symmetry properties of $D1$ in \eqref{yx} and \eqref{-y-x} imply that $D1$ satisfies the third boundary condition in \eqref{ANquadrantBCs}. We minimize the associated LdG energy functional in $\Omega_q$, given by:
\begin{gather}
J[P_{11},P_{12}]=\int_{\Omega_q \times[-1,0]} f_{el}(P_{11},P_{12})+f_b(P_{11},P_{12})\,\mathrm{dV}, 
\end{gather}
in the admissible space
\begin{gather}
\mathcal{A}_q:=\{(P_{11},P_{12})\in W^{1,2}(\Omega_q\times[-1,0];\mathbb{R}^2): \eqref{ANquadrantBCs}\ \textrm{is satisfied}\}.   
\end{gather}
As the boundary conditions on $\partial V_4$ are continuous and piecewise of class $C^1$, the admissible space, $\mathcal{A}_q$, is non-empty. Furthermore, $J$ is coercive on $\mathcal{A}_q$ and convex in the gradient $\nabla(P_{11},P_{12})$. Thus, by the direct method in the calculus of variations, we are guaranteed the existence of a minimizer $(P_{11}^*,P_{12}^*)\in\mathcal{A}_q$. We define a function $P_{11}^s\in V_4$ by even reflection of $P_{11}^*\in\Omega_q\times[-1,0]$ about the cross-section $z = 0$, and odd reflection of $P_{11}^*\in\Omega_q\times[-1,1]$ about the square diagonals. We do the same for the function $P_{12}^s\in V_4$ defined by odd reflections of $P_{12}^*$ about cross-section $z = 0$, and even reflections of $P_{12}^*$ about the square diagonals. The reflections across the mid-plane, $z=0$, gives us the $D2$ state for $z>0$, as required.
 By repeating the arguments in Lemma 2 and Lemma 3 of \cite{dangfifepeletier}, 
and Theorem 3 of \cite{dangfifepeletier}, 
 the constructed configuration, $(P_{11}^s, P_{12}^s)$, is a weak solution of the associated Euler-Lagrange equation on $V_4$. One can verify that $(P_{11}^s, P_{12}^s)$ is a critical point of $J$ on $\mathcal{A}_0$ with the desired properties.
On the middle cross-section $z=0$, $P_{12}(x,y,0)\equiv 0$, and $P_{11}(x,\pm x, 0) = 0$, so that we have two line-defects with $\P = 0$ on the square diagonals $y = \pm x$, and we have a $WORS$-like configuration on $z=0$, justifying the label, $D1 - WORS-D2$, for this state.
\end{proof}
 
\subsubsection{Instability for large $h$}\label{large_h}
In this section, the LdG energy functional is rescaled with the scaling $\hat{x} = \sqrt{2C/L}x$, $\hat{y} = \sqrt{2C/L}y$, $\hat{z} = \sqrt{2C/L}z/\lambda$,
\begin{equation}\label{eq:denormalized_xyz}
    F_{\lambda h}[\P] := \int_{V_{4\lambda h}} \left(\frac{1}{2}\left|\nabla_{xy}\P\right|^2 + \frac{1}{2\lambda^2}\left|\P_{,z}\right|^2 +  (-\frac{B^2}{8C^2} tr\P^2 + \frac{1}{8}( tr \P^2)^2)\right) \mathrm{dV},
\end{equation}
where $V_{4\lambda h} = E_{4\lambda}\times [-h,h]$ and $E_{4\lambda} = [-\lambda,\lambda]^2$.
The associated second variation of the rLdG energy at a critical point $p_c = (P_{11}^c,P_{12}^c)$, is given by:
\begin{equation}
    \partial^2F_{\lambda h}[\eta] = \int_{V_{4\lambda h}}|\nabla_{xy} \eta|^2 + \frac{1}{\lambda^2}|\eta_{,z}|^2 + \left(|p_c|^2-\frac{B^2}{4C^2}\right)|\eta|^2+2\left(p_c\cdot\eta\right)^2 \mathrm{dV},
    \label{eq:second_2}
\end{equation}
The stability of $\P_c$ is measured by the quantity
\begin{equation}
\mu_{\lambda}(h):=\inf_{\eta\in W^{1,2}_0(V_{4\lambda h})\backslash\{0\}}\frac{\partial^2 F_{\lambda h}[\eta]}{\int_{V_{4\lambda h}}\eta^2}.
\end{equation}

For 3D critical points of the rLdG energy in \eqref{eq:denormalized_xyz}, $p_c$, with $\partial_z p_c(x,y,\pm h) = 0$, 
we can compute an explicit upper bound for the second variation of the rLdG energy about $p_c$ as shown below.
The critical point of $\partial^2 F_{\lambda h}[\eta]$ is a solution of
\begin{equation}\label{eq:eta}
\Delta_{xy}\eta + \frac{1}{\lambda^2}\Delta_z\eta= (|p_c|^2-\frac{B^2}{4C^2})\eta + 2(p\cdot\eta)p.
\end{equation}
We set $\eta^* = p_{c,z}$, which vanishes on $\partial V_{4\lambda h}$ by assumption, and satisfies \eqref{eq:eta}, since $p_c$ satisfies the Euler-Lagrange equations
\begin{align}\label{eq:EL_lambda}
        \Delta_{xy}P_{11} + \frac{1}{\lambda^2}\Delta_z P_{11}&= \left(P_{11}^2+P_{12}^2-\frac{B^2}{4C^2}\right)P_{11}, \nonumber\\ 
        \Delta_{xy}P_{12} + \frac{1}{\lambda^2}\Delta_z P_{12}&= \left(P_{11}^2+P_{12}^2-\frac{B^2}{4C^2}\right)P_{12}.
\end{align}
Subsequently, the integral of the first and second terms in \eqref{eq:second_2} is 
\begin{align}
&\int_{V_{4\lambda h}} |\nabla_{xy} \eta^*|^2  + \frac{1}{\lambda^2}|\eta^*_{,z}|^2 dV = -(\int_{V_{4\lambda h}} \eta^*\Delta_{xy}\eta^*  + \frac{1}{\lambda^2}\eta^*\Delta_z\eta^* dV) = \int_{V_{4\lambda h}}-p_{c,z} (\Delta_{xy} p_c)_{,z} - \frac{1}{\lambda^2}p_{c,z} (\Delta_{z} p_c)_{,z}dV \nonumber\\
&= \int_{V_{4\lambda h}}-p_{c,z} (\left(|p_c|^2-\frac{B^2}{4C^2}\right)p_c)_{,z} dV \nonumber\\
&= \int_{V_{4\lambda h}}-2 (p_c\cdot p_{c,z})^2 - \left(|p_c|^2-\frac{B^2}{4C^2}\right)|p_{c,z}|^2 dV.
\end{align}
Substituting the above equations into \eqref{eq:second_2}, we have
\begin{equation}
\partial^2F_{\lambda h}[\eta^*] = 0.
\end{equation}
Therefore, for any $h$ and $\lambda$, 3D critical points, $p_c$, of \eqref{eq:denormalized_xyz}, that satisfy $p_{c,z}(x,y,\pm h) = 0$, are not stable since the smallest eigenvalue of the corresponding Hessian of the rLdG energy is non-positive, i.e.,
\begin{equation}
    \mu_{\lambda} (h) \leq 0.
\end{equation}

\begin{figure}
    \begin{center}
        \includegraphics[width=1.0\columnwidth]{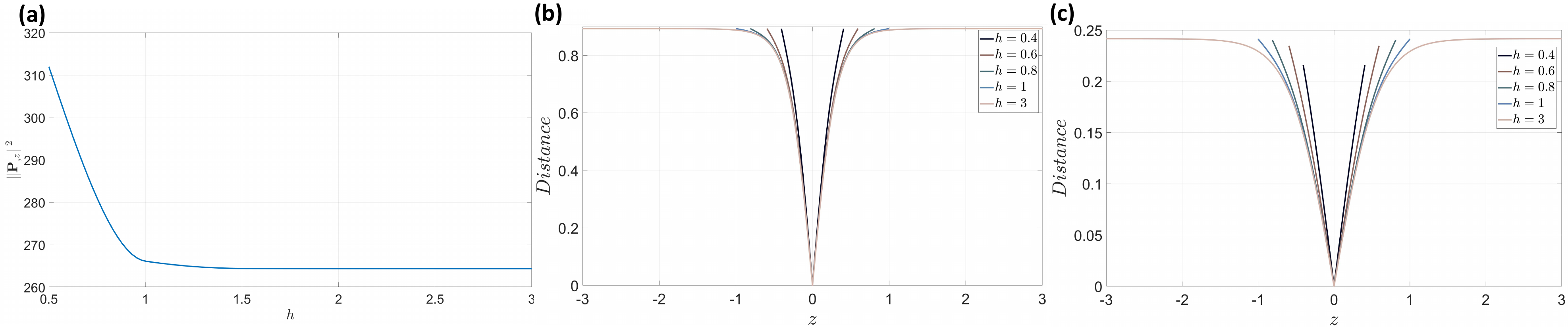}
        \caption{(a) The plot of $||\P_{,z}||^2_{L^2(V_{4\lambda h})}$ vs. $h$, where $\P$ is the critical point of the rLdG energy corresponding to \eqref{eq:denormalized_xyz}, $D1-WORS-D2$ with $\lambda^2 = 30$. In (b) and (c), we plot the scalar order parameter, $Distance = \sqrt{|P_{12}(x,y,z)^2 + 3(P_{11}(x,y,z)^2-P_{11}(x,y,0)^2)|}$, of the numerically computed $D1-WORS-D2$ state, at $(0,0,z)$ and $(0,\lambda/2,z)$, for different values of $h$.}
        \label{D1-WORS-D2}
    \end{center}
\end{figure}

In the following, we consider the stability of the $D1-WORS-D2$ critical point on $V_4$, for large $h$. 
As shown in Fig. \ref{D1-WORS-D2}(a), the integral of $|\partial_z p_c|^2$ on $V_{4\lambda h}$ is bounded, so we have $\partial_z p_c\to 0$ as $h\to\infty$, on the top and bottom surfaces and subsequently $\mu(\infty)\leq 0$. 
In Fig. \ref{D1-WORS-D2}(b), the order $\sqrt{P_{12}(x,y,z)^2 + 3(P_{11}(x,y,z)^2-P_{11}(x,y,0))}$ drops to zero at $(x,y,z) = (0,0,0)$, consistent with the $WORS$-like state with $P_{12}(x,y,0)\equiv 0$ on $z=0$. We observe in Fig. \ref{D1-WORS-D2}(b) and (c) that as $h$ increases, the effect of the middle $WORS$ slice remains confined to a thin transition layer confined to a small neighbourhood of $z=0$, whilst the $D1-WORS-D2$ solution approaches a block-like structure outside this transition layer. Namely, the solution is effectively the $\P^b = D1$ state for $z\in [-h, a)$, a transition state mediated by a $WORS$-type profile for $z\in [a,b]$ followed by the $\P^t = D2$ state for $z\in (b, h]$, for some fixed $a, b$ independent of $h$, as $h\to\infty$. We also note that the $D1$ and $D2$ states, with no interior defects, are more energetically favourable and more stable than the middle $WORS$-slice, which has two diagonal defect lines.

In order to use the proof for the instability of the 2D WORS (for $\lambda$ large) in \cite{schatzman1995stability} and \cite{canevari2017order}, we rotate the square by $45$ degrees, so that $\bar{\P}$ is related to $\P$ by
\begin{equation}\label{rotation}
    \left(\begin{tabular}{cc}
        $\bar{P}_{11}$ & $\bar{P}_{12}$\\
        $\bar{P}_{12}$ & $-\bar{P}_{11}$
    \end{tabular}\right)(\r) = S\P(S^T\r)S^T =
    \left(\begin{tabular}{cc}
        $-P_{12}$ & $P_{11}$\\
        $P_{11}$ & $P_{12}$
    \end{tabular}\right)
    (S^T\r)
\end{equation}
where $S$ is the corresponding rotation matrix. Hence, the condition
$P_{12}(x,y,0)\equiv 0$, for $(x,y)\in [-\lambda,\lambda]^2$ proved in Proposition \ref{prop:forever_critical} translates to $\bar{P}_{11}(x,y,0)=0$, $(x,y)\in S[-\lambda,\lambda]^2$ , and the condition, $P_{11}(x,\pm x,z)\equiv 0$, $x\in[-\lambda,\lambda]$ proved in Proposition \ref{prop:forever_critical} \ translates to $\bar{P}_{12}(x,0,z)\equiv 0$, for $x\in[-\sqrt{2}\lambda,\sqrt{2}\lambda], z\in [-h, h]$ and $\bar{P}_{12}(0,y,z)\equiv 0$ $y\in[-\sqrt{2}\lambda,\sqrt{2}\lambda], z\in[-h,h]$. In the following Remark, we omit the bars over $P_{11}$ and $P_{12}$ for brevity.

\begin{remark}\label{prop:unstable}
Assuming, $\sqrt{|P_{11}(x,y,\sigma)^2+3(P_{12}(x,y,\sigma)^2-P_{12}(x,y,0)^2)|} \sim O(\sigma)$ and 
$xyP_{12,zz}(x,y,0)\leq 0$, for $\lambda\gtrsim 1/\epsilon^2$ with small constant $\epsilon$, if the $D1-WORS-D2$ critical point has the multi-block structure as described above, and $h$ is large enough, it is strictly unstable in the sense that the Hessian of the rLdG energy \eqref{p_energy} about this critical point has a negative eigenvalue. 
\end{remark}
This is not a rigorous proof, but rather a set of heuristic arguments based on numerical estimates. However, it gives a clear physical interpretation of the origin of the instability of the $D1-WORS-D2$ critical point -  an instability localised near the centre of the cuboid that stems from the instability of the $WORS$ on a square domain, for $\lambda$ large enough and for $h$ large enough.

The critical point is $D1-WORS-D2$, labelled by $p = (P_{11},P_{12})$, which is a solution of
\eqref{eq:EL_lambda}.
Consider the second variation of the rLdG energy in \eqref{eq:second_2} and we construct a perturbation $\eta = (\eta_1,\eta_2)$, such that the associated second variation is negative. 
We work with the perturbation
\begin{equation}\label{eq:pert}
\eta(x,y,z) = \tilde{\eta}(x,y)\phi(z),
\end{equation}
where $\phi\in C_0^{\infty}(\mathbb{R})$ is a cutoff function satisfying $|\phi'|\leq C/\sigma$ which is equal to $1$ over $[-\sigma/2,\sigma/2]$ and vanishes outside $[-\sigma,\sigma]$, where $\sigma$ is less than the thickness of jump layer, the region with $s<\frac{1}{2}$ in Fig. \ref{D1-WORS-D2}. 

The 2D perturbation $\tilde{\eta}(x,y)$ is an unstable direction for a 2D WORS solution $\tilde{p} = (\tilde{P}_{11},\tilde{P}_{12})$, which is reported in Lemma 3.4 of \cite{schatzman1995stability}. The components, $\tilde{P}_{11}$ and $\tilde{P}_{12}$, satisfy
\begin{align}
&\tilde{P}_{11}\equiv 0,\ on\ SE_{4\lambda}\\
&\Delta_{xy}\tilde{P}_{12} =  (\tilde{P}_{12}^2-\frac{B^2}{4C^2})\tilde{P}_{12}\ on\ SE_{4\lambda}.
\end{align}
As for the perturbation $\tilde{\eta}$ in \cite{schatzman1995stability}, we assume
\begin{equation}
\tilde{\eta}(x,y) =
\begin{cases}
\psi(x/n)P^{\infty}_{12,y}(|x|,y),\ \text{if}\ |x|\geq \epsilon,\\
P^{\infty}_{12,y}(\epsilon,y)+P^{\infty}_{12,xy}(\epsilon,y)\frac{x^2-\epsilon^2}{2\epsilon},\ \text{if}\ |x|\leq\epsilon.
\end{cases}
\end{equation}
where $\mathbf{P}^{\infty}$ satisfies
\begin{align}
&P^{\infty}_{11}\equiv 0,\ on\ \mathbb{R}^2\\
&\Delta_{xy}P^{\infty}_{12} =  ({P^{\infty}_{12}}^2-\frac{B^2}{4C^2})P^{\infty}_{12}\ on\ \mathbb{R}^2.
\end{align}
$\psi\in C_0^{\infty}(\mathbb{R})$ is a cutoff function, which is equal to $1$ over $[-1,1]$ and vanishes outside $[-2,2]$, and $n$ is a large positive number.

We substitute the perturbation in \eqref{eq:pert}, into the second variation \eqref{eq:second_2},
\begin{align}
    \partial^2F_{\lambda h}[\eta] &= \int_{SV_{4\lambda h}}|\nabla_{xy} \tilde{\eta}|^2\phi^2 + \frac{1}{\lambda^2}(\tilde{\eta}\phi')^2 + \left(P_{11}^2 + 3P_{12}^2-\frac{B^2}{4C^2}\right)(\tilde{\eta}\phi)^2 dV\\
    &= \int_{SV_{4\lambda h}}(|\nabla_{xy} \tilde{\eta}|^2 + \left( 3\tilde{P}_{12}^2-\frac{B^2}{4C^2}\right)\tilde{\eta}^2)\phi^2 dV\\
    &+ \int_{SV_{4\lambda h}} \frac{1}{\lambda^2}(\tilde{\eta}\phi')^2 + \left(P_{11}^2 + 3(P_{12}^2-\tilde{P}_{12}^2)\right)\tilde{\eta}^2\phi^2 dV\\
    &= \int_{SV_{4\lambda h}}(|\nabla_{xy} \tilde{\eta}|^2 + \left( 3\tilde{P}_{12}^2-\frac{B^2}{4C^2}\right)\tilde{\eta}^2)\phi^2 dV \label{eq:2D}\\
    &+ \int_{SV_{4\lambda h}} \frac{1}{\lambda^2}(\tilde{\eta}\phi')^2 + \left(P_{11}^2 + 3(P_{12}^2-P_{12}(x,y,0)^2\right)\tilde{\eta}^2\phi^2 dV\label{eq:z}\\
    &+ \int_{SV_{4\lambda h}}  3\left(P_{12}(x,y,0)^2-\tilde{P}_{12}^2\right)\tilde{\eta}^2\phi^2 dV,
    \label{eq:h}
\end{align}
where $SV_{4\lambda h}$ is the 3D cuboid domain after rotation.
According to the stability analysis of 2D WORS in Lemma 5.5 of \cite{canevari2017order} and Lemma 3.4 of \cite{schatzman1995stability}, the order of the integral in \eqref{eq:2D} is $O(-\epsilon\sigma) + O(\sigma n^{-1})$.
From Corollary 2.9 in \cite{schatzman1995stability}, we infer that for all $\mu\in(0,\lambda_1)$ ($\lambda_1$ is the smallest positive eigenvalue of the spectrum of $H$, where $HP^{\infty}_{12} =
-\Delta_{xy}P^{\infty}_{12} + \left({P^{\infty}_{12}}^2-\frac{B^2}{4C^2}\right)P^{\infty}_{12}$), there exists a constant $c(\mu)$, such that 
\begin{gather}
|P^{\infty}_{12,y}(x,y)|\leq c(\mu)e^{-|y|\mu},\ 
|P^{\infty}_{12,xy}(x,y)|\leq c(\mu)\min(e^{-|x|\mu},e^{-|y|\mu}).
\end{gather}

Hence 
\begin{align}
    &\int_{SE_{4\lambda}} \tilde{\eta}^2 dxdy = \int_{SE_{4\lambda}\cap{|x|\leq\epsilon}} \tilde{\eta}^2 dxdy + \int_{SE_{4\lambda}\cap{|x|\geq\epsilon}} \tilde{\eta}^2 dxdy\\
    \leq
    &\int_{SE_{4\lambda}\cap{|x|\leq\epsilon}} \left(P^{\infty}_{12,y}(\epsilon,y)+P^{\infty}_{12,xy}(\epsilon,y)\frac{x^2-\epsilon^2}{2\epsilon}\right)^2 dxdy+
    \int_{SE_{4\lambda}\cap{|x|\geq\epsilon}} \psi^2(x/n){P^{\infty}}^2_{12,y}(|x|,y) dxdy\\
    \leq& \int_{-\lambda}^{\lambda} 4c^2(\mu)e^{-2|y|\mu}2\epsilon dy + 
    \int_{-2n}^{2n}\int_{-\lambda}^{\lambda} c^2(\mu)e^{-2|y|\mu} dydx\\
    =& O(n),
\end{align}
where $SE_{4\lambda}$ is the $E_{4\lambda}$ after rotation.
The first term in \eqref{eq:z} is 
\begin{equation}
    \int_{SV_{4\lambda h}} \frac{1}{\lambda^2}(\tilde{\eta}\phi')^2 dV
    = \frac{1}{\lambda^2}\int_{SE_{4\lambda}} \tilde{\eta}^2 dxdy \int_{[-h,h]}(\phi')^2 dz
    = O(n\sigma^{-1} \lambda^{-2}).
\end{equation}

The most drastic changes of the nematic order, along the $z$-axis, happen on $(x,y) = (0,0)$. From our numerical result (see Fig \ref{D1-WORS-D2}(b) and (c)), when $h$ is large enough, the term $\sqrt{|P_{11}(x,y,z)^2+3(P_{12}(x,y,z)^2-P_{12}(x,y,0)^2)|}$ is almost linear in $|z|$ near the center and thus can be controlled by $O(\sigma^2)$. 
The second term in \eqref{eq:z} can be controlled by $O(\sigma^3 n)$.

\begin{figure}
    \begin{center}
        \includegraphics[width=0.3\columnwidth]{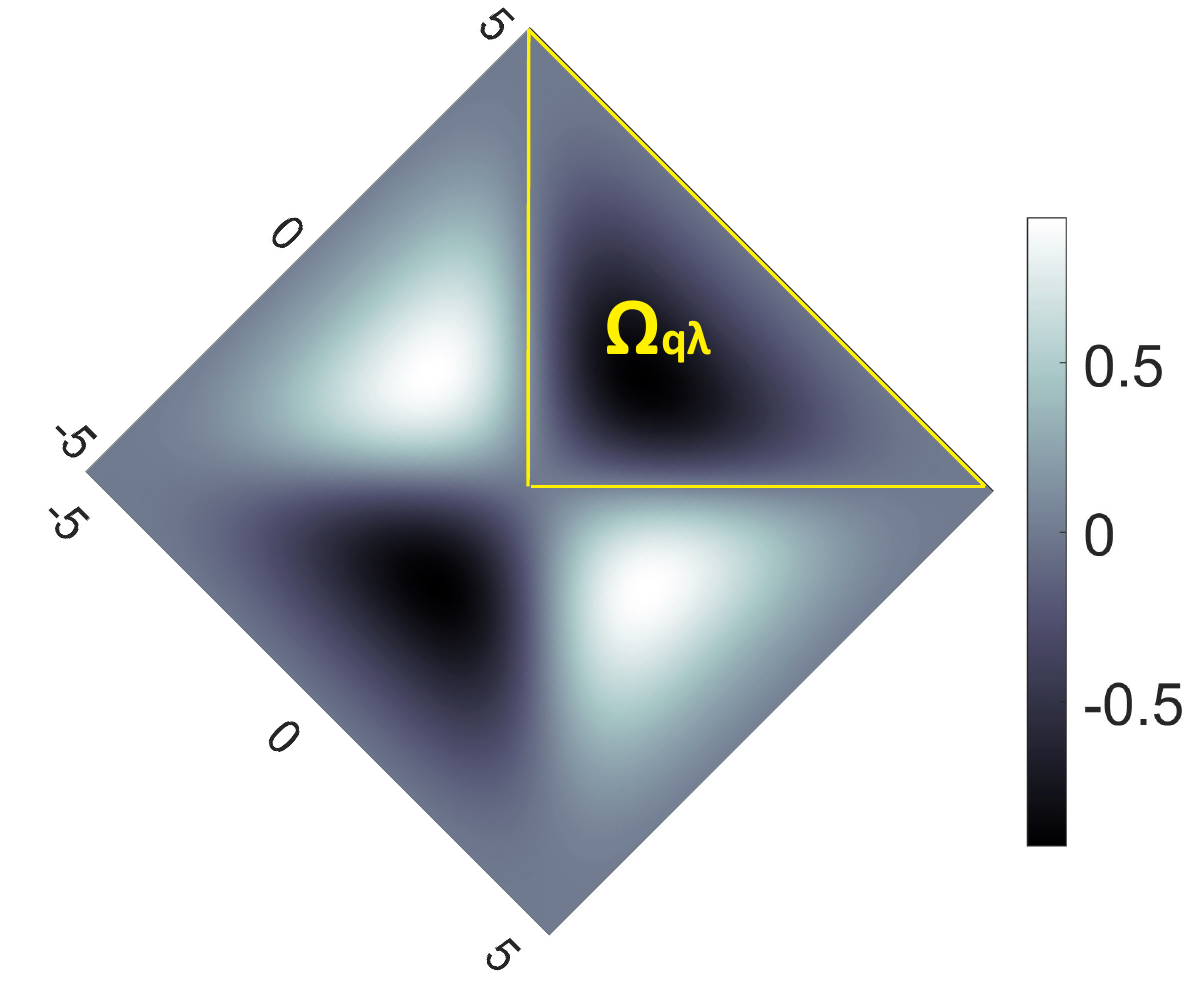}
        \caption{The plot of $P_{12,zz}$ on $SE_{4\lambda}$The area of domain $\Omega_{q \lambda}$ is framed with yellow lines.}
        \label{p12zz}
    \end{center}
\end{figure}
Let's define $\Omega_{q\lambda} := SE_{4\lambda}\cap \{x\geq 0\}\cap\{y\geq 0\}$.
The 2D WORS solution $\tilde{P}_{12}(x,y)$ satisfies
\begin{align}
&-\Delta_{xy} \tilde{P}_{12} + (\tilde{P}_{12}^2 - \frac{B^2}{4C^2})\tilde{P}_{12}^2 = 0\ on\ \Omega_{q\lambda} \label{eq:p}\\
&\tilde{P}_{12} = \frac{B}{2C}\ on\ \partial\Omega_{q\lambda}\backslash\{x=0\}\backslash\{y = 0\},\\
&\tilde{P}_{12} = 0\ on\ x = 0\ and\ y = 0.
\end{align}
According to Fig. \ref{p12zz}, for D-WORS-D, $P_{12,zz}(x,y,0)\leq 0$ on $\Omega_q$.
Hence $P_{12}(x,y,0)$ is a subsolution of \eqref{eq:p}.
Then $p_1 := max\{P_{12}(x,y,0),\tilde{P}_{12}(x,y)\}$ is a subsolution of \eqref{eq:p}.
The constant $B/2C$ is a supersolution of \eqref{eq:p}.
Therefore, by the classical sub- and supersolution method (Theorem 1 p.508 in \cite{evans2022partial}), there exists a solution $p_2$ of \eqref{eq:p} such that $p_1\leq p_2 \leq B/2C$.
Due to the maximum principle, we have $\tilde{P}_{12}(x,y)\geq 0$ on $\Omega_{q\lambda}$, and subsequently $p_2\geq 0$.
According to Lemma 4.2 in \cite{canevari2017order}, there is a unique non-negative solution of \eqref{eq:p}. Thus, we deduce $P_{12}(x,y,0)\leq p_1\leq p_2 = \tilde{P}_{12}(x,y)$, i.e., $P_{12}^2(x,y,0)-\tilde{P}^2_{12}(x,y)\leq 0$ on $\Omega_{q\lambda}$. One can repeat the arguments above on the remaining three quadrants to deduce that $P_{12}^2(x,y,0)-\tilde{P}^2_{12}(x,y)\leq 0$ on $SE_{4\lambda}$.

Finally, we have 
\begin{equation}
    \partial^2 F_{\lambda h}[\eta]\leq O(-\epsilon \sigma + \sigma n^{-1} + n\sigma^{-1}\lambda^{-2} + \sigma^3n)
\end{equation}
and the second variation is negative $\partial^2 F_{\lambda h}[\eta]\leq 0$
when $n \gtrsim \epsilon^{-1}$, $\sigma  \lesssim \sqrt{\epsilon/n}  \lesssim \epsilon$, $\lambda\gtrsim\sqrt{n/\epsilon}\sigma^{-1}\gtrsim n/\epsilon\gtrsim \epsilon^{-2}$, i.e. $\sigma$ is small enough; $h$ and $\lambda$ are large enough.

\begin{figure}
    \begin{center}
        \includegraphics[width=0.7\columnwidth]{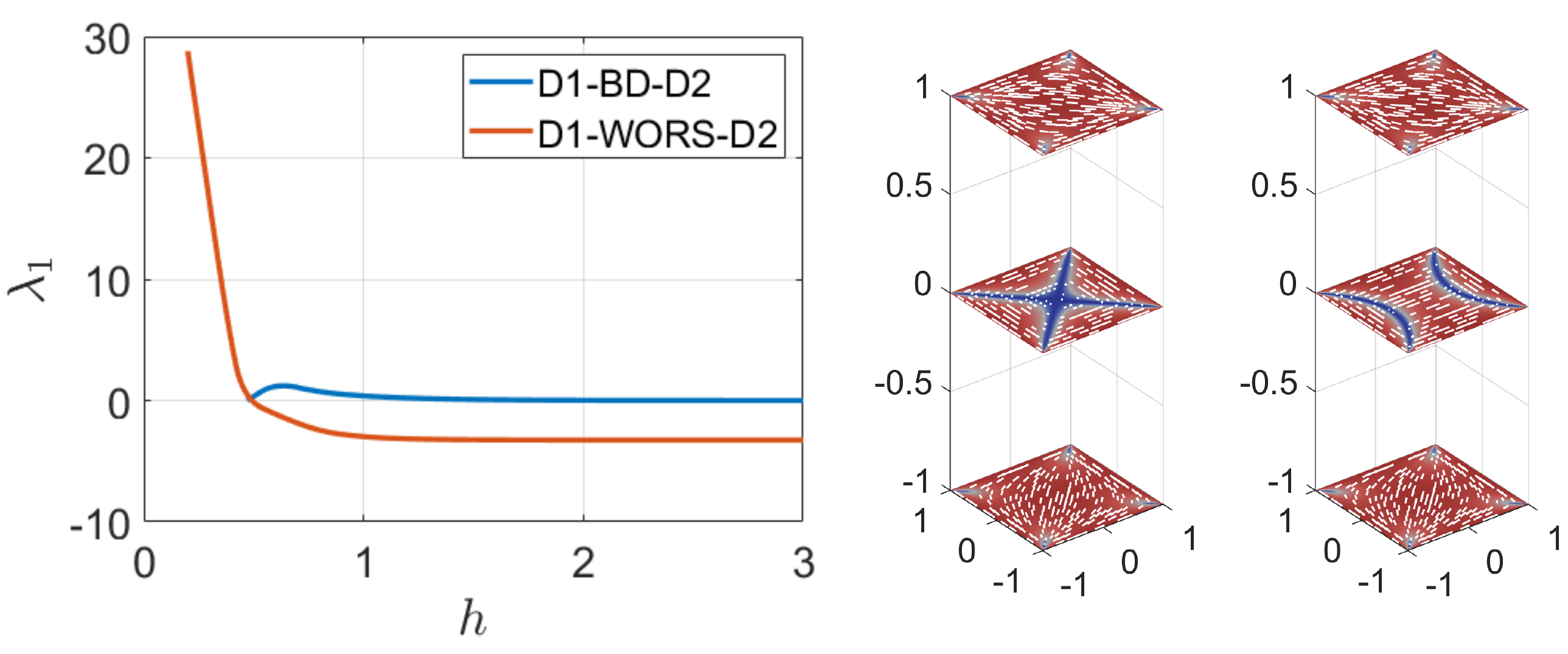}
        \caption{The smallest eigenvalue, $\lambda_1$, of the Hessian of  the critical points, $D1-WORS-D2$ and $D1-BD-D2$, for $\lambda^2=30$ v.s. $h$, and the plots of these two critical points at $h=1,\lambda^2=30$.}
        \label{lambda_1}
    \end{center}
\end{figure}

In Fig. \ref{lambda_1}, we plot the bifurcation diagram for solutions of \eqref{eq:EL} in a prism $V_4$ v.s. the height $h$. For $h$ small enough, the $D1-WORS-D2$ critical point is the unique stable state. As $h$ increases, the $D1-WORS-D2$ loses stability and bifurcates into two critical points: $D1-BD1-D2$ and $D1-BD2-D2$. The $BD$ states are unstable $z$-independent critical points of \eqref{p_energy} on $V_{4}$, subject to $\P = \P_l$ on the lateral surfaces. This is an interesting observation that we can observe unstable 2D states, such as $WORS$ and $BD$, (which are expected to be difficult to observe in purely 2D situations) by imposing stable boundary conditions, $D1$ and $D2$, on the top and bottom surfaces of the cuboid or $V_4$. $BD1$ and $BD2$ are two $BD$ critical points, related by a rotation, and are energetically degenerate. The smallest eigenvalue of the Hessian of the rLdG energy evaluated at the two distinct states, $D1-BD1-D2$ and $D1-BD2-D2$, are the same. An analogous bifurcation diagram has been reported in  \cite{shi2022hierarchies,canevari2020well}, with Neumann boundary conditions on the top and bottom surfaces of $V_4$.
The $D1-BD1-D2$ and $D1-BD2-D2$ critical points  also have a multi-block structure, with the $BD$-states confined to a small layer localised near $z=0$. 

As $h\to\infty$, according to Remark \ref{prop:unstable},  the smallest eigenvalue of the $D1-WORS-D2$ critical point converges to a negative constant (see Fig. \ref{lambda_1}). The eigenvector corresponding to this negative eigenvalue changes the $WORS$-type configuration on $z=0$, to a $BD$-type critical point on $z=0$. The $D1-BD-D2$ has a zero smallest eigenvalue for the associated Hessian of the rLdG energy in \eqref{p_energy} (see Fig. \ref{lambda_1}). The intuitive explanation is that the eigenvector corresponding to the zero eigenvalue moves the transition layer around $z=0$ up or down (provided it remains sufficiently far from the top and bottom surfaces), without changing the energy. Hence $D1-BD-D2$ cannot be strictly stable. 
This is consistent with the computation at the beginning of this subsection which demonstrates that $\mu(\infty)\leq 0$ for any multi-block critical point, $p_c$ such that $p_{c,z}\to 0$ as $h\to\infty$, on the top and bottom surfaces.

Whilst studying $z$-independent critical points of the rLdG energy \eqref{p_energy} on the square domain, subject to the boundary conditions $\P = \P_l$ on the square edges, the $WORS$ is index-$4$, and the $BD$ is index-$2$ for $\lambda^2 = 30$. There are other unstable states too, such as the index-$3$ $T$ and index-$2$ $H$ critical points (see Fig. \ref{TT}(b)). Fixing $D1$ and $D2$ to be $\P^b$ and $\P^t$ respectively, only the $BD$ and $WORS$-critical points appear on $z=0$, for the mixed 3D critical points. Hence, the boundary conditions impose constraints on the state observed on and around $z=0$, and consequently, the mixed 3D critical points, index-0 $D1-BD-D2$ and index-1 $D1-WORS-D2$ on a 3D cuboid have lower indices than the $BD$ and $WORS$ on a 2D square, respectively. We conjecture that the index of a 2D unstable critical point, $p_u$, is always higher than the index of a mixed LdG critical point on a 3D prism, which exhibits a $p_u$-type interior profile connecting the Dirichlet boundary conditions, $\P^t$ and $\P^b$ respectively.

\subsubsection{The $\lambda\to\infty$ limit}
In the $\lambda\to\infty$ limit, rLdG minimisers converge to minimisers of the bulk energy i.e. if $\P$ is a rLdG minimiser for large enough $\lambda$, then $|\P|^2 \to \frac{s_+^2}{2}$ , where $s_+ = \frac{B}{C}$, at least everywhere away from the edges and vertices of the cuboid geometry and defects. The choice of $s_+$ is dictated by the temperature, in this case $A = -\frac{B^2}{3C}$. This can be seen informally by the competition between the bulk and elastic energy terms in \eqref{p_energy} in the $\lambda \to \infty$ limit, and rigorously using variational arguments as in \cite{majumdar2010landau} for example. In other words, to leading order, in the $\lambda \to \infty$ limit, the rLdG minimiser is of the form
$$
\P = s_+(\n\otimes\n-\I/2),
$$
away from the vertices and edges of $V_K$,  where 
$\mathbf{n} = (\cos\theta,\sin\theta)$.
Hence, the energy functional in \eqref{p_energy} reduces to 
\begin{equation}\label{infty_energy}
F_{\infty} = \int_{V_4} |\nabla_{xy} \theta|^2+\frac{1}{2h^2}\left|\theta_{,z}\right|^2\textrm{dV}.
\end{equation}
The Dirichlet boundary conditions on the top and bottom surfaces, $\P^{t}$ and $\P^b$, have directors, $\n^{t} = (\cos \theta^{t},\sin \theta^{t})$ and $\n^{b} = (\cos \theta^{b},\sin \theta^{b})$ respectively. Hence, to leading order, for the rLdG minimiser, $\theta$ is a solution of
\begin{equation}\label{eq:EL_infty}
        \Delta_{xy}\theta + \frac{1}{h^2}\Delta_z\theta = 0,
\end{equation} subject to the boundary conditions, $\theta = \theta^t$ on $z=1$, and $\theta = \theta^b$ on $z=-1$.
If $\P^t$ and $\P^b$ are chosen to be stable 2D critical points of the rLdG energy i.e. stable solutions of \eqref{eq:top_bottom}, then for $\lambda$ sufficiently large, $\theta^t$ and $\theta^b$ are (to leading order) solutions of the Laplace equation (see \cite{han2020reduced} for more details).
Given these conditions on $\P^t$ and $\P^b$, one can easily check that the corresponding solution of \eqref{eq:EL_infty} is linear in $z$ and is given by,
\begin{equation}\label{eq:solution_infty}
\theta(x,y,z) = \frac{z+1}{2} \theta^{t}(x,y) + \frac{(1-z)}{2}\theta^{b}(x,y).
\end{equation}
The leading order rLdG minimiser is given by, $\P=s_+(\n\otimes\n-\I/2)$ where $\n = (\cos\theta,\sin\theta)$ and $\theta$ is given by \eqref{eq:solution_infty}, in the $\lambda\to\infty$ limit. The solution in \eqref{eq:solution_infty} is a good approximation to minimisers of the reduced LdG energy, for large enough $\lambda$, see Fig. \ref{lambda_infty}. 

Next, we give some examples of solutions of the form of \eqref{eq:solution_infty}.
We divide the domain $\Omega$ into four triangular domains, $\Omega_i$ by the diagonal lines $x = \pm y$ (see Fig. \ref{divide_domain}). We denote the two diagonal legs of $\Omega_i$ as, $C_{il}$ and $C_{ir}$, respectively. This domain division is only useful/applicable for $V_4$ with $D1$ and $D2$ as $\P^b$ and $\P^t$ respectively, since the numerically computed rLdG critical points only exhibit defects along the diagonals and edges of $E_4$.
\begin{figure}
    \begin{center}
        \includegraphics[width=0.4\columnwidth]{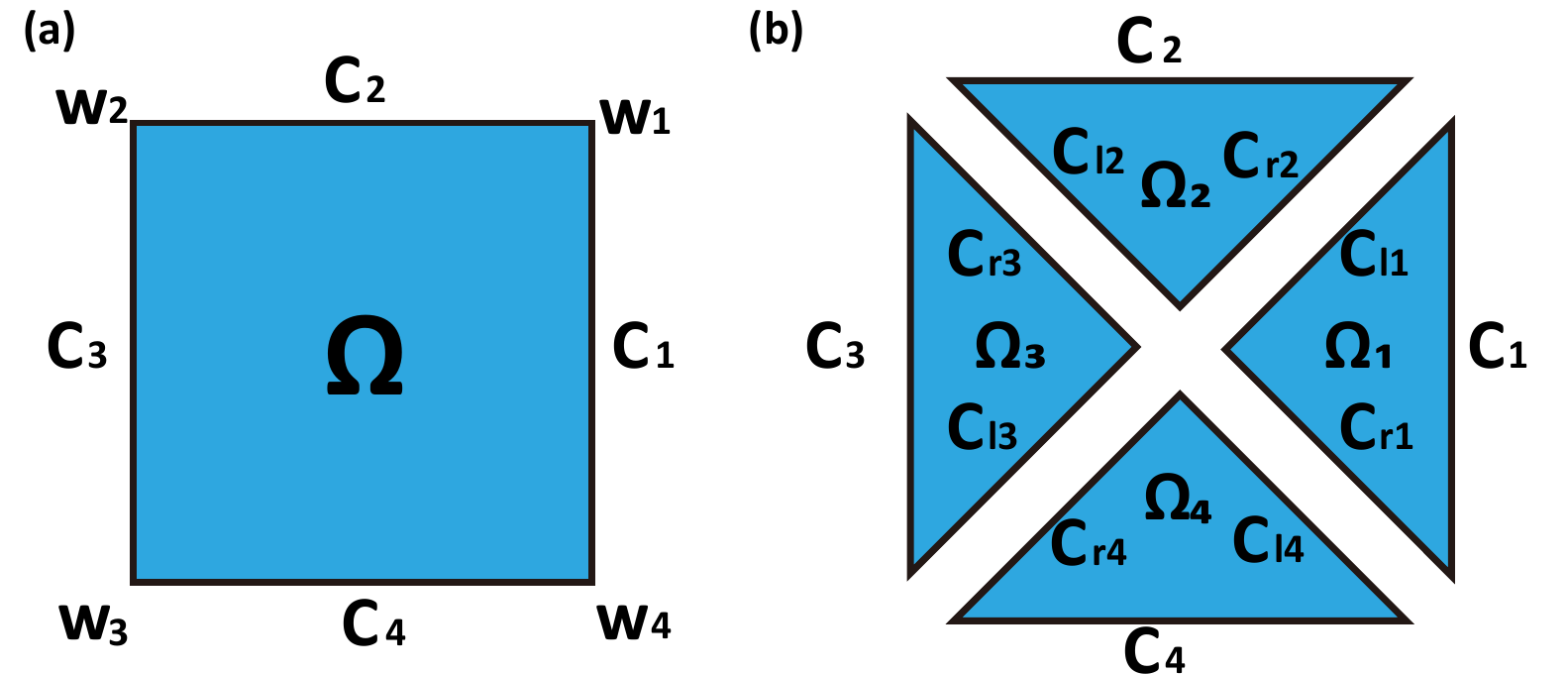}
        \caption{We divide the square domain, $\Omega$, into four sub-domains, $\Omega_i$, with diagonal legs $C_{ri}$ and $C_{lr}$, $i = 1,\cdots,4$.}
        \label{divide_domain}
    \end{center}
\end{figure}

\begin{figure}
    \begin{center}
        \includegraphics[width=0.6\columnwidth]{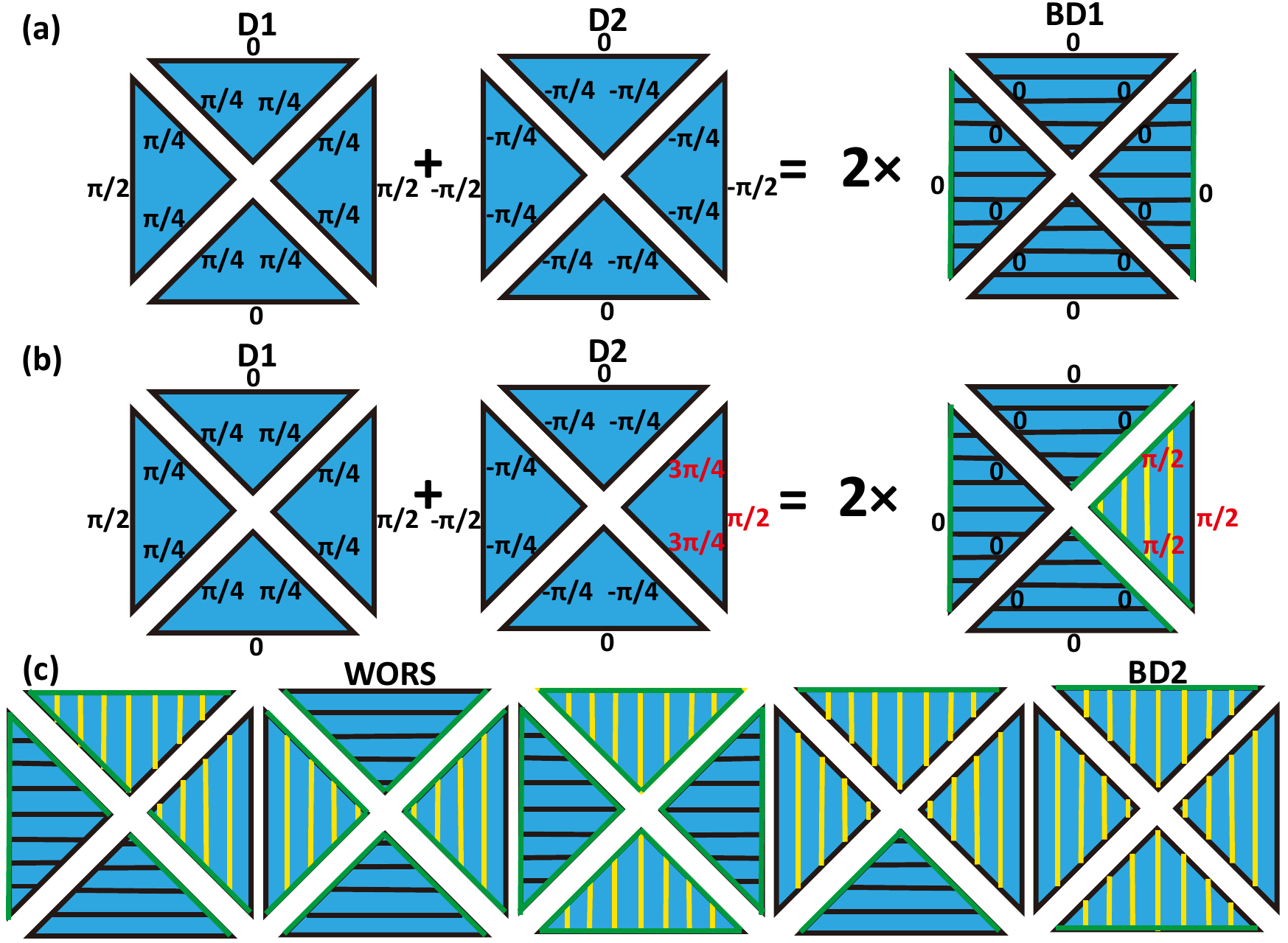}
        \caption{(a) and (b) illustrate the different possibilities for the middle layer configuration on $z=0$, obtained from different choices of the boundary conditions, $\theta^b$ and $\theta^t$ on $C_{il}$ and $C_{ir}$ respectively (corresponding to $D1$ and $D2$ on the bottom and top of $V_4$). (c) Other potential choices for the configuration on the middle layer. Black/yellow lines indicates that the nematic director is along $x/y$-axis. The $\theta^{D2}$($\theta^t$) is increased by $\pi$ on the $\Omega_i$ occupied by yellow lines. Green lines on edges or diagonals represents line defects.}
        \label{fig:possibility}
    \end{center}
\end{figure}
We have multiple choices for $\theta^t$($\theta^b$), corresponding to the same state. The boundary conditions, $\theta^t$ and $\theta^b$ can jump by multiples of $\pi$ across common diagonal edges of adjacent triangular domains. 
In the simplest case, we prescribe $\theta^b$ and $\theta^t$ on the diagonal edges of $\Omega_i$, without any discontinuities or jumps. For example, for the $D1$ and $D2$ solutions, we can have
\begin{gather}\label{eq:D_theta}
\theta^{D1}_b = \begin{cases}
\pi/2\ on\ C_1\ and\ C_3,\\
0\ on\ C_2\ and\ C_4,\\
\pi/4\ on\ C_{ri}\ and\ C_{li},\ i = 1,\cdots, 4.
\end{cases}
\theta^{D2}_t = \begin{cases}
-\pi/2\ on\ C_1\ and\ C_3,\\
0\ on\ C_2\ and\ C_4,\\
-\pi/4\ on\ C_{ri}\ and\ C_{li},\ i = 1,\cdots,4.\\
\end{cases}
\end{gather}
Hence on $z=0$, $\theta$ is given by (refer to \eqref{eq:solution_infty})
\begin{equation}
\theta = (\theta^{D1}+\theta^{D2})/2 \equiv 0,\ on\ \Omega,\\
\end{equation}
which is the $BD1$ state (Fig. \ref{fig:possibility}(a)) that has two line defects on the left and right edges. This is the candidate for the global rLdG minimiser in the $\lambda \to \infty$ limit.

We fix the boundary conditions for $D1$ ($\theta^b$) as in \eqref{eq:D_theta}, and then increase $\theta^t$ on one or more of the diagonal edges of $\Omega_i$ by $\pi$, i.e., increase $\theta$ on the middle layer by $\pi/2$, to generate more candidates for the middle layer configuration on $z=0$. Given $\theta^t$ and $\theta^b$, we get the solution $\theta$ by \eqref{eq:solution_infty} on the square quadrant (bounded by two diagonal legs and one square edge).
This allows us to define $\P^{t}$  by the relations, $P_{11}^{t} = \frac{s_+}{\sqrt{2}}\cos 2\theta^t, P_{12}^{t} = \frac{s_+}{\sqrt{2}}\sin 2 \theta^{t}$. We then have $\P^t$ and $\P^b$ on the square quadrant, and then define a critical point on the entire prism, by reflecting the solution on the quadrant as in Proposition~\ref{prop:forever_critical}. The computed rLdG tensors, $\P$ are not expected to be approximate rLdG minimisers for large enough $\lambda$, but could be good initial conditions for computing unstable saddle points of the rLdG energy in this limit. 
For example, we modify $\theta^t$ as shown below 
\begin{equation}
\theta^{D2}_t = \begin{cases}
\pi/2\ on\ C_1\\
-\pi/2\ on\ C_3,\\
0\ on\ C_2\ and\ C_4,\\
3\pi/4\ on\ C_{r1}\ and\ C_{l1},\\
-\pi/4\ on\ C_{ri}\ and\ C_{li},\ i = 2,\cdots,4.\\
\end{cases}
\end{equation}
Subsequently, we have the following profile for $\theta$ on $z=0$ (refer to \eqref{eq:solution_infty}): 
\begin{equation}
\theta = (\theta^{D1}+\theta^{D2})/2=
\begin{cases}
 \pi/2,\ on\ \Omega_1,\\
 0,\ on\ \Omega_i,\ i = 2,\cdots,4,\\
\end{cases}
\end{equation}
accompanied by line defects on $\{(x,y):y = \pm x,x\geq 0\}$ (Fig. \ref{fig:possibility}(b)). 

We can generate more possibilities for $\theta$ on $z=0$ by adding multiples of $\pi$ to $\theta^t$, on the diagonal edges of $\Omega_i$ (see Fig. \ref{fig:possibility}(c)). These modified boundary conditions generate line defects and asymmetric configurations, and are unlikely to be observed in practice. We only observe the $BD1$, $BD2$ and $WORS$ configurations on $z=0$, for the numerically computed critical points of the rLdG energy, with $(\P^b, \P^t) = (D1, D2)$. 




\begin{figure}
    \begin{center}
        \includegraphics[width=0.7\columnwidth]{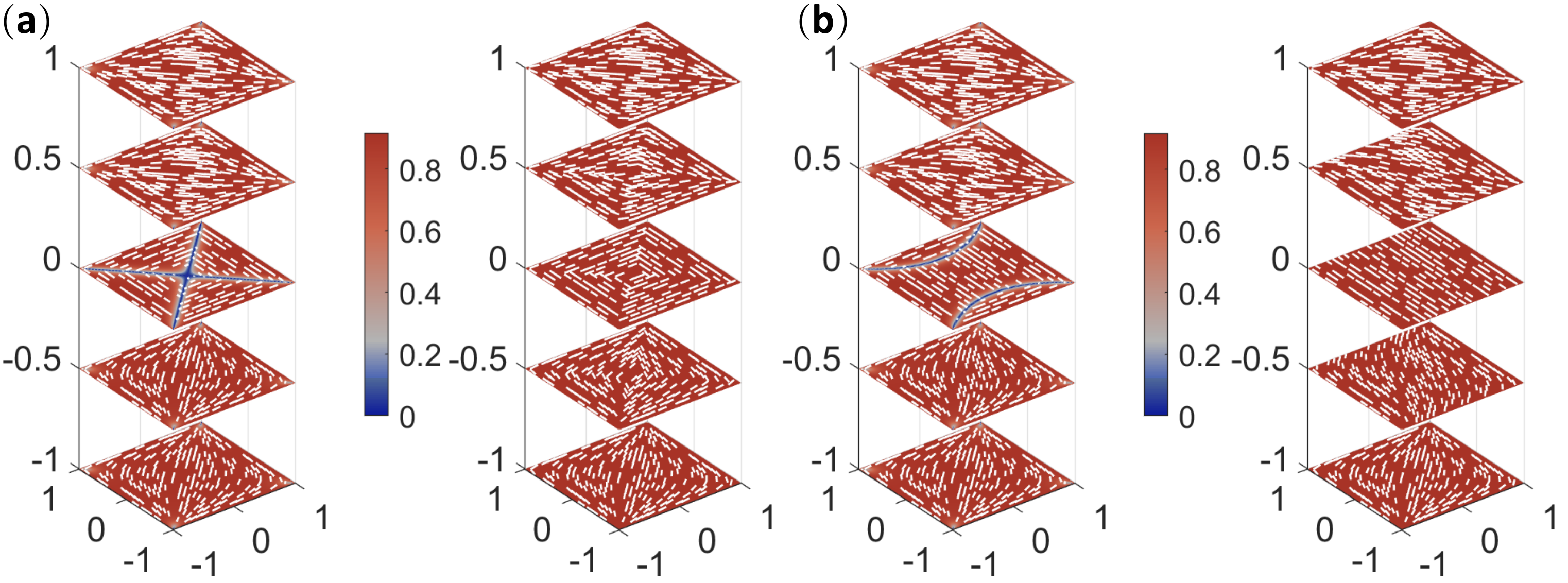}
        \caption{The plots of the solutions of Euler--Lagrange equation in \eqref{eq:EL} with $\lambda^2 = 300$ and the corresponding limiting profile as $\lambda\to\infty$, with $h = 1$, (a) $D-WORS-D$, (b) $D-BD-D$. In (a) and (b), we plot the numerical solutions on the left for which the nematic director, given by the vector field $(\cos(arctan(P_{12}/P_{11})/2),\sin(arctan(P_{12}/P_{11})/2))$ is plotted by white lines, and the order parameter, $\sqrt{P_{11}^2 + P_{12}^2}$, is plotted in terms of the color bar. The right images (for (a) and (b)) are the limiting profiles, the white lines plot the vector field, $(\cos(2\theta),\sin(2\theta))$, and the red color implies that the order parameter is constant, $s \equiv \frac{B}{C}$, for the limiting profile.} 
        \label{lambda_infty}
    \end{center}
\end{figure}
\subsection{$(\P^b, \P^t)$: $T1$ and $T2$}
The unstable states, $T1$ and $T2$, exist on a 2D square domain for $\lambda^2\geq 25$, with a line defect along one of the square diagonals (Fig. \ref{2D_solution}). For $\lambda^2 = 30$, the critical point index-$3$ $T$ is connected with a higher-index critical point, index-$4$ $WORS$, and lower-index critical points, index-$2$ $BD$ and $H$, index-$1$ $J$, index-$0$ $D$ and $R$ (Fig. \ref{TT}(b)) \cite{yin2020construction}. We can use this information to construct families of mixed 3D critical points, with the fixed boundary conditions on the top and bottom surfaces of $V_4$. In what follows, we label critical points by means of strings e.g. $A - B - C - D$ where $A$ and $D$ are the fixed boundary conditions, and $B$ and $C$ are $z$-invariant solutions of the reduced Euler-Lagrange equations, compatible with the lateral boundary conditions, $\P^l$.
The complicated bifurcation diagram is partially illustrated in Fig. \ref{TT}(d), where the corresponding states are plotted in Fig. \ref{TT}(c). 

For $h$ small enough, the $T-WORS-T$ critical point (with the $WORS$ configuration on $z=0$) is the unique stable state. As $h$ increases, $T-WORS-T$ loses stability and bifurcates into a stable $T-D-T$ and an index-$1$ $T-R-BD-R-T$; this is analogous to the bifurcation from a stable $WORS$ to a stable $D$ solution and unstable $BD$ solution on a square domain as the edge length increases, see \cite{robinson2017molecular}. Then the $T-WORS-T$ further bifurcates to an index-$2$ $T-R-T-R-T$, and an index-$3$ $T-D-WORS-D-T$. 
The index-$1$ $T-R-BD-R-T$ critical point further bifurcates into an index-$2$ $T-R-BD-R-T$ critical point, and an index-$1$ $T-R-BD-R-T(2)$. 
The index-$3$, $T-D-WORS-D-T$ critical point, further bifurcates into an index-$2$ $T-D-WORS-D-T$ and an index-$3$ $T-J-T-J-T$. The state, $T-D-WORS-D-T$, contains the familiar $D-WORS-D$, as discussed in Section \ref{sec:cuboid}. 
One can find the corresponding 2D pathways on the 2D solution landscape in Fig. \ref{TT}(b), which give rise to these exotic 3D critical points of the rLdG energy with fixed initial and end points (boundary conditions). For example, the pathways in Fig. \ref{TT}(a), $T\rightarrow D\rightarrow WORS \rightarrow D \rightarrow T$ correspond to the 3D solution $T-D-WORS-D-T$ in Fig. \ref{TT}(b) and $T\rightarrow R\rightarrow BD \rightarrow R \rightarrow T$ corresponds to the 3D solution, $T-R-BD-R-T$ in Fig. \ref{TT}(b). For $(\P^b, \P^t) = (D1, D2)$, the mixed 3D critical points correspond to pathways on the 2D solution landscape in Fig. \ref{TT}(a) via one higher-index states(index-2 $BD$/index-4 WORS in 2D). For $(\P^b, \P^t) = (T1, T2)$, we obtain 3D critical points that correspond to different kinds of pathways on the 2D solution landscape i.e. 2D pathways via lower-index states (index-0 $D$/index-0 $R$/index-1 $J$)
and/or  a higher-index state (index-2 $BD$/index-3 $T$/ index-4 $WORS$) like $T-D-WORS-D-T$, or via a single higher-index saddle point as in the $T-WORS-T$ critical point.  
\begin{figure}
    \begin{center}
        \includegraphics[width=0.8\columnwidth]{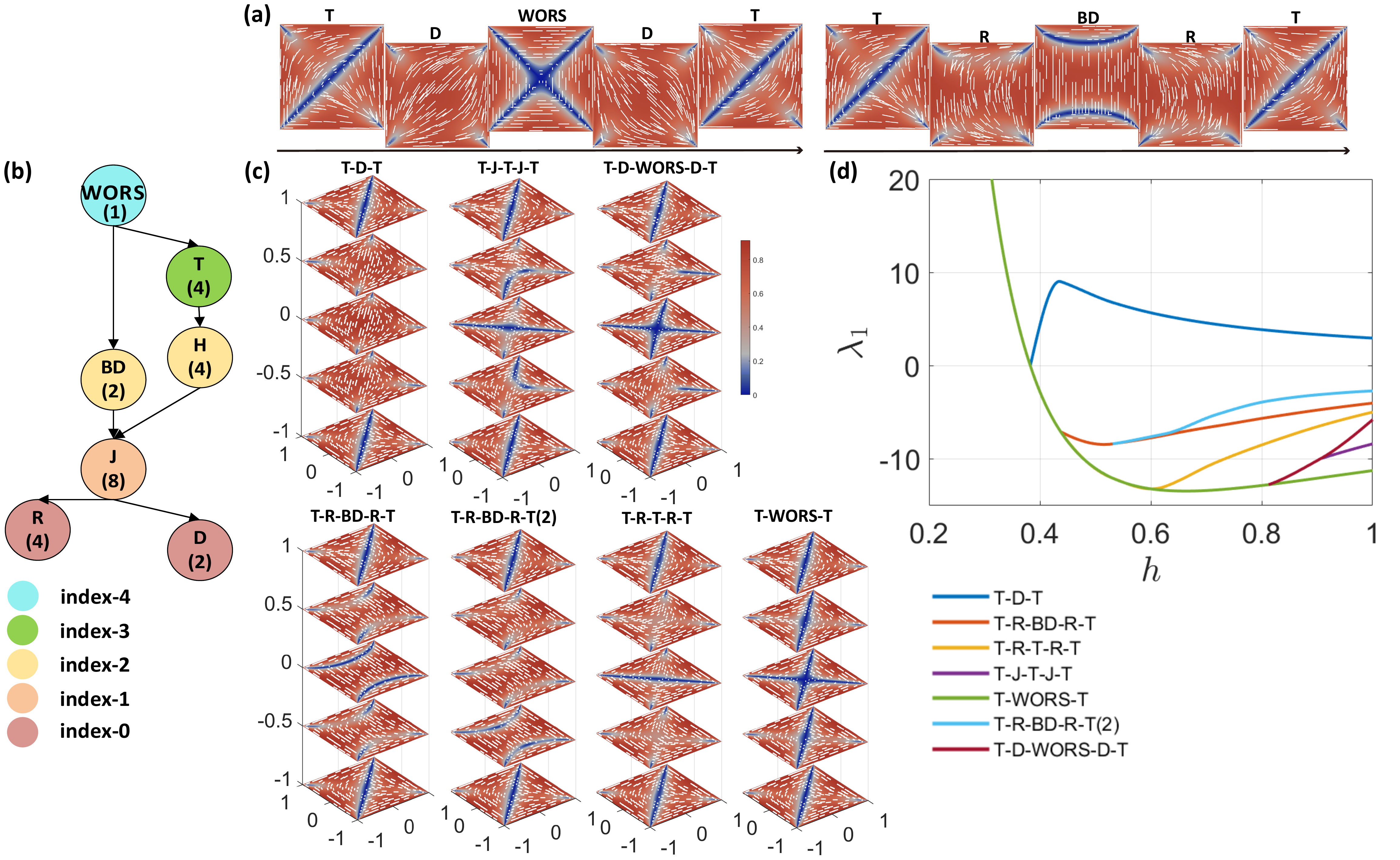}
        \caption{
        (a) Two 2D pathways between the two fixed $T$ solutions on the top and bottom,  which correspond to 3D solutions, $T-D-WORS-D-T$ and $T-R-BD-R-T$. The $T$, $BD$ and the $WORS$-profiles are unstable saddle points of the rLdG energy on $E_4$, whilst $D$ and $R$ solutions are stable solutions of \eqref{eq:top_bottom} with $\lambda^2=30$. (b) The solution landscape with 2D critical states from index-4 to index-0. The number in each disk indicate the number of states in the class. (c) The 3D solutions of \eqref{eq:EL} with $(\P^b, \P^t) = (T1, T2)$ with $\lambda^2 =30$ and $h = 1$. The blue color label the point defect or line defect. (d) The smallest eigenvalue $\lambda_1$ of the multiple critical points. In this and all the subsequent figures, the color bar encodes the order parameter $\sqrt{P_{11}^2+ P_{12}^2}$, and the white lines label the nematic director, $(cos(arctan(P_{12}/P_{11})/2),sin(arctan(P_{12}/P_{11})/2))$.
        }
        \label{TT}
    \end{center}
\end{figure}

In contrast to the case with stable $D1$ and $D2$ solutions as boundary conditions, we get a genuine rLdG minimiser $T-D-T$ as $h\to \infty$, for which the second variation of the rLdG energy is strictly positive. The $T$ states are unstable and have higher energy than the middle $D$ layer. Consequently, the middle layer extends to a small neighbourhood of the boundaries, and we get thin transition layers near the boundaries, $z=\pm h$. Hence, letting $\P_c = T-D-T$, we do not have $\P_{c,z} \to 0$ as $z\to \pm h$, and  there is no eigenvector corresponding to the movement of jump layers or transition layers, i.e., we lose the analogy of the zero-eigenvalue in $D-BD-D$. 
\section{Other prism $V_K$}
\label{sec:prism}

We systematically study various classes of stable solutions on 2D polygons, in the rLdG framework, for large $\lambda$ in \cite{han2020reduced}. Namely, we use simple combinatorial arguments to show that there are at least $\frac{K(K-1)}{2}$ stable rLdG equilibria on a regular $K$-polygon with $K$-edges, and $[\frac{K}{2}]$ classes of stable equilibria, not related by rotation and reflection. The stable rLdG equilibria are distinguished by the location of the splay vertices i.e. the polygon vertices for which the nematic director rotates by $2\pi/K-\pi$ around the vertex. We make certain physically reasonable assumptions about the boundary data, to show that the stable states always have two splay vertices  and the distinct equilibria classes are distinguished by the relative location of two splay vertices \cite{han2020reduced}. For example, on 2D hexagon, there are three classes of stable equilibria: $Para$ with a pair of diagonally opposite splay vertices, $Meta$ with a pair of splay vertices separated by one vertex, and $Ortho$ with two adjacent splay vertices. 

Next, we adapt the 2D arguments in \cite{han2020reduced} to make some elementary predictions about the number of mixed 3D critical points on the prism, $V_K$, that has a polygonal cross-section, $E_K$, with $K$-edges. For example, take $\P^t$ and $\P^b$ to be two distinct global energy minimisers on $E_K$ 
i.e. a \emph{Para} state with two diagonally opposite splay vertices on a hexagon and a \emph{Meta} state, with two splay vertices separated by a vertex, on a pentagon. 
Since the Laplace operator in \eqref{eq:EL} is rotationally invariant, we can rotate the regular $K$-polygon domain around the $z$-axis so that the boundary conditions on the top and bottom surfaces have the following reflection symmetry property about $y = 0$ axis, 
\begin{align}
P_{11}^t(x,y) &= P_{11}^b(x,-y),\\
P_{12}^t(x,y) &= -P_{12}^b(x,-y).
\end{align}
For small $h$, there is a unique rLdG minimiser (critical point) for a given $\P^t$ and $\P^b$, and we can show that if (recall the arguments in \ref{small_h}) $(P_{11},P_{12})(x,y,z)$ is a solution of \eqref{eq:EL}, then $(P_{11},-P_{12})(x,-y,-z)$ is also a solution of \eqref{eq:EL}, subject to the boundary conditions above. The solution is unique for small $h$ and hence, has the symmetry property 
\begin{align}
P_{11}(x,y,z) = P_{11}(x,-y,-z),\\
P_{12}(x,y,z) = -P_{12}(x,-y,-z).
\end{align}
On the middle cross-section $z=0$, the reflection symmetry axis is $y = 0$, and we have $P_{12}(x,0,0) = -P_{12}(x,0,0) = 0$, i.e., the nematic director is either parallel or perpendicular to the symmetry axis, or we have a defect with $\P = 0$ on $y=z=0$. 

To distinguish between the three situations, the sign of $P_{11}$ is taken into consideration. If $P_{11}(x,0,0) \geq/\leq/= 0$, then the nematic director is parallel/ perpendicular/undefined along the symmetry axis $(y,z) = (0,0)$, the third case corresponding to a nematic defect. The three situations are captured by three commonly observed 2D solutions - the $WORS$ on a square domain, with line defects along $(y,z) = (0,0)$; the $BD1$ on a pentagon domain (Fig. \ref{fig:pentagon}), with nematic director perpendicular to the symmetry axis, $(y,z) = (0,0)$; and the $BD2$ on a pentagon domain (Fig. \ref{fig:pentagon}), with the nematic director parallel to $(y,z) = (0,0)$ almost everywhere in the interior of the pentagon. 
The $WORS$ is exclusive to a square domain, but the $BD$-states are generic. We conjecture that for given $\P^t \neq \P^b$ in the same class of globally stable rLdG equilibria, the $\P^b - BD - \P^t$ is the unique rLdG minimiser for small $h$, that remains potentially stable for all $h$ (or at least has non-negative second variation), on a generic prism $V_K$ for $K>4$.  


Repeating the same arguments as in \cite{han2020reduced}, we argue that for a given $\P^t$ and $\P^b$ on $V_K$, where $\P^t$ and $\P^b$ belong to the same class of lowest energy rLdG equilibria on $E_K$, 
(e.g. two $Meta$ states on $E_5$, and two $Para$ states on $E_6$),  there are $[K/2]$ ($[K/4]$) distinct classes of mixed 3D critical points on $V_K$ with odd(even) $K$, not related by rotation and reflection. For example on $V_5$, we have $[5/2] = 2$ classes of mixed 3D solutions $Meta-BD1-Meta$ and $Meta-BD2-Meta$ in Fig. \ref{fig:pentagon}. On $V_6$, we have $[6/4] = 1$ class of 3D mixed solution, labelled by $Para-BD-Para$ in Fig. \ref{fig:para}. 

\begin{figure}
    \begin{center}
        \includegraphics[width=0.3\columnwidth]{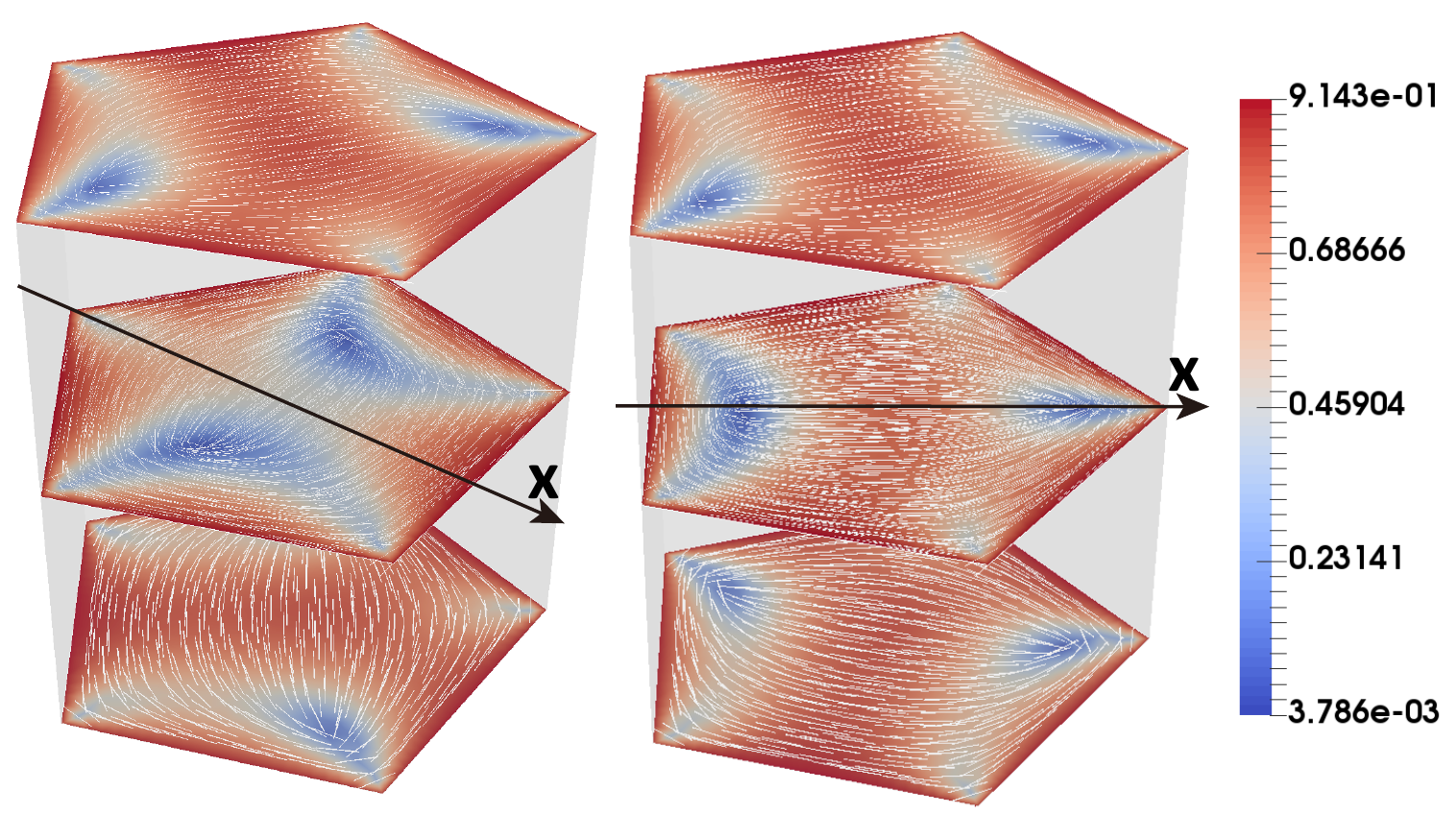}
        \caption{The profiles of mixed solutions of \eqref{eq:EL} in $V_5$ with two $Meta$ on the top and bottom,  $Meta-BD1-Meta$ and $Meta-BD2-Meta$  at $\lambda^2=30$ and $h = 0.3$. The right state has lower energy than the left. The axis indicates the symmetry axis $(y,z) = (0,0)$.}
        \label{fig:pentagon}
    \end{center}
\end{figure}

\subsection{Hexagonal prism, $V_6$}

In this section, we restrict attention to a hexagonal prism $V_6$, since we study the 2D problem on $E_6$ in detail in \cite{han2021solution,han2020reduced}. On $E_6$, there are three competing stable classes of rLdG equilibria - $Para$, $Meta$, and $Ortho$ with two splay vertices each, amongst which $Para$ has the lowest energy, for large enough $\lambda$.
Additionally, in \cite{han2021solution}, a new stable state $TRI$ is reported, for $\lambda$ large enough. $TRI$ has one central $-1/2$ point defect and three splay vertices, and has the highest energy among all the numerically computed stable rLdG equilibria on $E_6$. We numerically compute a plethora of unstable index-$k$ saddle points of the rLdG energy on $E_6$, for which the Hessian of the rLdG energy has $k$ negative eigenvalues or $k$ unstable directions. We report pathways between the stable rLdG equilibria on $E_6$, and there are pathways mediated by the commonly reported transition states (index-$1$ saddle points) and pathways mediated by high-index saddle points in $E_6$, illustrating the fascinating connectivity of the 2D solution landscape of the rLdG model on $E_6$ in \cite{han2021solution}.

Consider $V_6$ and fix $\P^t$ and $\P^b$ to be two different $Para$ states, with diagonally opposite splay vertices. By the arguments in the preceding sections, for $h$ small enough, we expect to find a $P-BD-P$ state with a $BD$-like profile on $z=0$; indeed we numerically find a $P-BD-P$ state (see Fig. \ref{fig:para}(a)). In Fig. \ref{fig:para}(a), we notice blue tubes connecting the splay vertices of $\P^t$ and $\P^b$, and the blue tubes are defect lines running through the height of $V_6$. These defect lines could have pronounced optical and mechanical responses in experiments. The mixed 3D critical point, $P- BD-P$ corresponds to a pathway between two $Para$($P$) states on the 2D solution landscape, which proceeds via a $BD$ state. 
The $P-BD-P$ state is always a minimiser and doesn't bifurcate into other critical points. As $h\to\infty$, the smallest eigenvalue of the Hessian of the rLdG energy about the $P-BD-P$ state tends to be zero (Fig. \ref{fig:para}(c)), which is analogous to $D-BD-D$ state in Section \ref{sec:cuboid}.

Next, we fix $\P^t$ and $\P^b$ to be two different $TRI$ states, on top and bottom surfaces. Here the solution landscape is rich and we find multiple stable and unstable mixed 3D critical points (see Fig. \ref{fig:TRI}). All the mixed critical points exhibit line defects (blue tubes) running across the height of $V_6$, and these line defects connect the splay vertices and the central point defects on the top and bottom. There are multiple combinations of these line defects, which offer multiple possibilities for exotic morphologies. Similar line defects have been observed in a 3D cylinder, where there are straight defect lines  and defect rings in both experiments and numerical simulations \cite{williams1972nonsingular, han2019transition}.
We may not have found all the mixed 3D critical points with these fixed boundary conditions, but it is notable that the numerically computed mixed 3D critical points have corresponding counterpart pathways on the 2D solution landscapes reported in \cite{han2021solution}. For example, in \cite{han2021solution}, we report a pathway between the fixed $TRI$ states constructed by four transition pathways, via index-$1$ states, $T0$ and $M1$. We numerically find the mixed 3D critical point $TRI-T0-M-M1-P-M1-M-T0-TRI$, with the stable $Para$ state in the middle, corresponding to the 2D pathway reported in \cite{han2021solution} (see Fig. \ref{fig:TRI}(a)).
We also report pathways via high-index saddle points like $BD$ in \cite{han2021solution}.
Fig. \ref{fig:TRI}(b) shows a part of the solution landscape where $TRI$ is directly or indirectly connected via high-index saddle points. In Fig. \ref{fig:TRI}(c), from bottom to the top of $TRI-T135-RING-T135-TRI$ state, the three $+1/2$ defect near vertices move towards the central $-1/2$ defect, merge together and we obtain the Ring solution with a unique central $+1$ defect, and then reverse the process to connect to $TRI$($\P^t$).
Again, this mixed 3D critical point corresponds to a pathway between $\P^t$ and $\P^b$ on the 2D solution landscape.
The mixed critical point, $TRI-T0-M-BD-M-T0-TRI$ is constructed by two transition pathways, via index-$1$ $T0$ state and one pathway via the index-2 $BD$ state. The mixed states, $TRI-T130-P-T130-TRI$ and $TRI-T10-P-T10-TRI$ go through two index-3 $T130$ states, and an index-2 $T10$ state, respectively. This illustrates the relevance of unstable higher index saddle points on $E_6$, for rLdG critical points in three dimensions. 
From Fig. \ref{fig:TRI}(d), and according to our numerical computations, 
when $h$ is small, $TRI-T135-Ring-T135-TRI$ is stable, and for $h$ large enough, both $TRI-T0-M-M1-P-M1-M-T0-TRI$ and $TRI-T0-M-M1-BD-M1-M-T0-TRI$ are stable.

\begin{figure}
    \begin{center}
        \includegraphics[width=0.7\columnwidth]{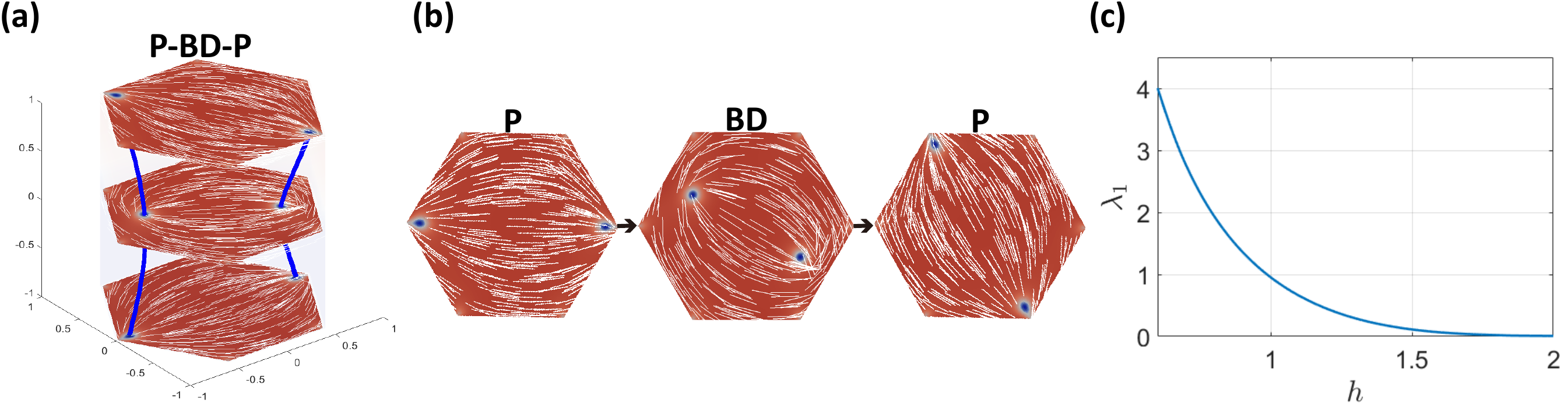}
        \caption{(a) The 3D stable state $P-BD-P$, the solution of \eqref{eq:EL} with $\lambda^2=600$ and $h = 1$; (b) the 2D pathway between two distinct $Para$ via a $BD$ with $\lambda^2 = 600$; (c) the smallest eigenvalue of $P-BD-P$.}
        \label{fig:para}
    \end{center}
\end{figure}

\begin{figure}
    \begin{center}
        \includegraphics[width=\columnwidth]{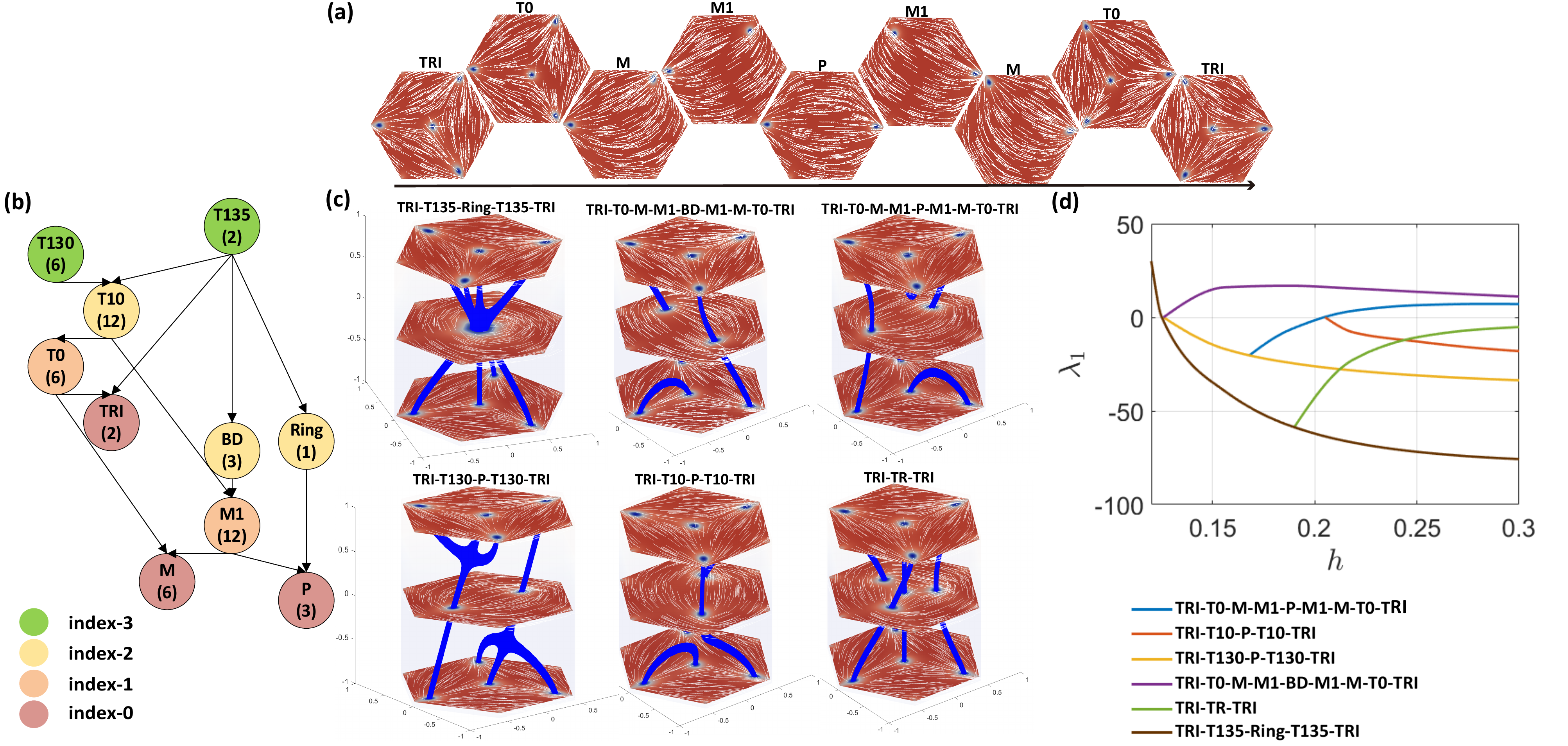}
        \caption{(a) A 2D transition pathway between two $TRI$ solutions, corresponding to the 3D solution $TRI-T0-M-M1-P-M1-M-T0-TRI$. The profiles on the top line are index-1 saddle points, and on the bottom line are minimizers, the solutions of \eqref{eq:top_bottom} with $\lambda^2=600$. (b) A part of solution landscape showing the connectivity of index-3 saddle points to index-0 saddle points. The number in each disk indicate the number of states (related by rotation and reflection) in the class. (c) The 3D solutions of \eqref{eq:EL} with two $TRI$ states as $(\P^b, \P^t)$, with $h = 0.4$, $\lambda^2 = 300$. The blue lines are the defect lines, for which $|\P|$ approaches zero or is of much smaller magnitude than the neighbouring region. (d) The smallest eigenvalue $\lambda_1$ of the Hessian of the rLdG energy evaluated about the multiple critical points.}
        \label{fig:TRI}
    \end{center}
\end{figure}
\begin{figure}
    \begin{center}
        \includegraphics[width=0.8\columnwidth]{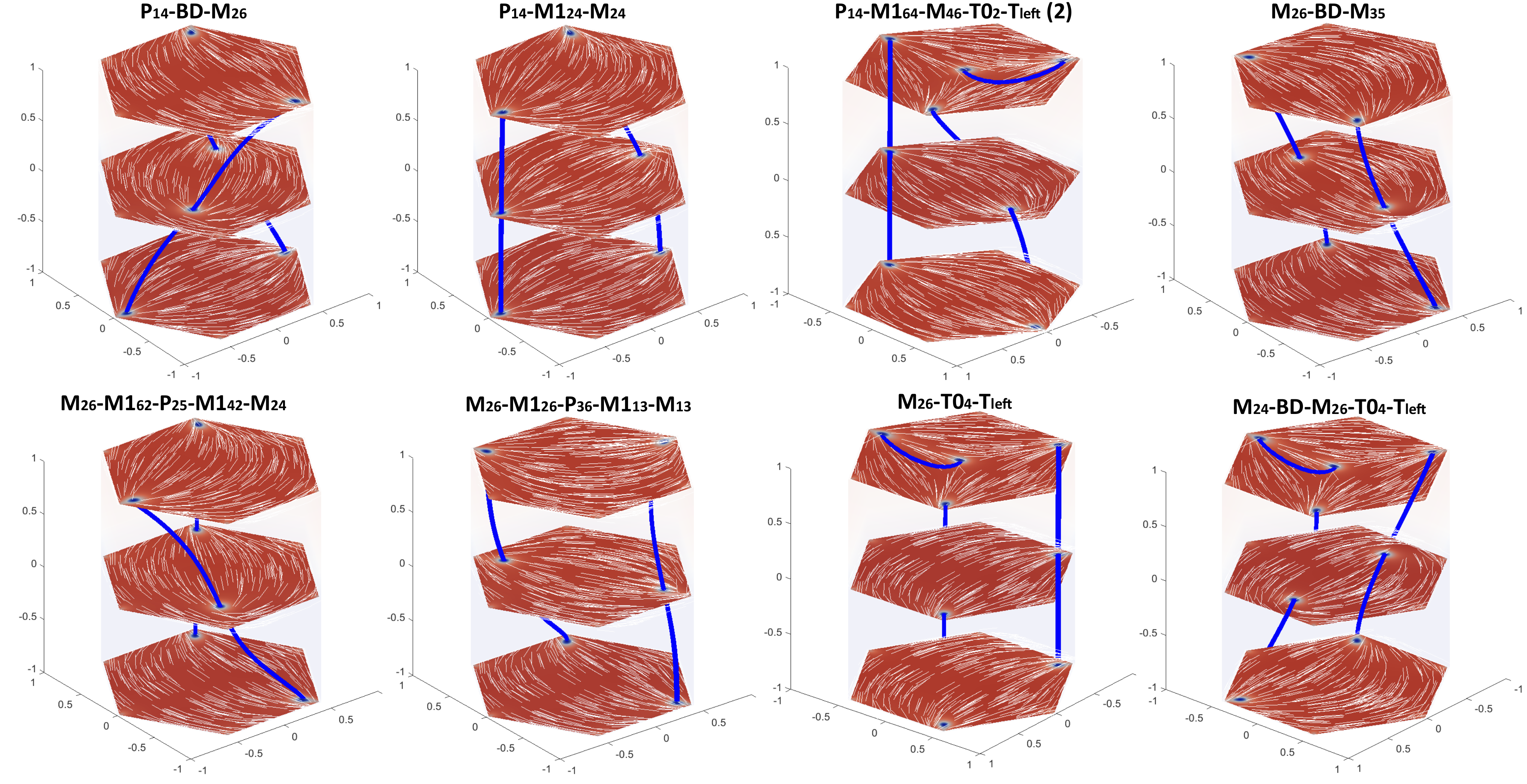}
        \caption{The unique 3D solutions of \eqref{eq:EL} for which $\P^b$ and $\P^t$ belong to different equivalence classes of solutions of \eqref{eq:top_bottom}, for $\lambda^2=600$ and $h = 0.1$. The subscripts indicate the location of  defects at the vertices or in the interior. Readers are referred to \cite{han2021solution} for nomenclature details.}
        \label{fig:mixed}
    \end{center}
\end{figure}
We can generate other classes of mixed 3D critical points by choosing $\P^t$ and $\P^b$ to belong to different solution classes e.g. $Para$ on either $z=\pm 1$ accompanied by $Meta$ or $TRI$ on the opposite boundary surface (see  Fig. \ref{fig:mixed}). We do not make definite conclusions since the solution landscape is hugely complex, but these numerical examples do demonstrate the tremendous possibility of multistability in 3D, generated by tessellating 2D solutions of the rLdG model or stacking 2D critical points on top of each other, and the sequence of the 2D critical points seems to be intimately connected to corresponding pathways between $\P^t$ and $\P^b$ on the 2D solution landscapes. 


\section{Conclusions}
\label{sec:conclusions}
In this paper, we study rLdG equilibria (solutions of \eqref{eq:EL} or equivalently critical points of \eqref{p_energy}) on the prism $V_K$, with a regular polygon cross-section $E_K$ with $K$ edges, and fixed Dirichlet boundary conditions on the top and bottom surfaces. We work in a re-scaled framework and there are two key parameters - a dimensionless parameter $\lambda$ which incorporates the cross-section edge length, and the parameter, $h$, which is a measure of the prism height. Our framework is comparable to that studied in \cite{shi2022hierarchies}, where we fix the boundary conditions on the lateral surfaces of $V_K$ but impose Neumann boundary conditions on the top and bottom surfaces. In this paper, we choose the Dirichlet conditions on the top and bottom surfaces, $\P^t$ and $\P^b$, to be $z$-invariant solutions of \eqref{eq:EL}, consistent with the lateral boundary conditions. In contrast, in \cite{shi2022hierarchies}, $\P^t$ and $\P^b$ are determined as part of the energy minimisation process, and both approaches have scientific and practical value.

For a given $\P^t$ and $\P^b$, we prove that there exists a unique rLdG energy minimiser in our admissible space, for $h$ sufficiently small. If $\lambda$ is large enough so as to allow for different classes of $z$-invariant solutions, then we take $\P^t \neq \P^b$ and search for mixed 3D critical points. It is difficult to perform exhaustive asymptotic studies as $h\to 0$, $h\to \infty$, or as $\lambda \to \infty$ in 3D, as in \cite{han2020reduced}, but we obtain some analytic insights accompanied by illuminating numerical results. We consider $V_4$ in some detail, with a square cross-section, and two different choices of $(\P^b, \P^t)$. For the first example, we take $\P^t$ and $\P^b$ to be two different diagonal solutions, $D1$ and $D2$, both of which are $z$-invariant local minimisers of \eqref{p_energy}. We prove the existence of a $D1-WORS-D2$ critical point on $V_4$ for all $h$, such that there are two defect lines along the square diagonals on $z=0$. For $h$ small enough, this critical point is globally stable and in fact, the unique critical point, and loses stability as $h$ increases. As $h$ increases, the $D1-WORS-D2$ state bifurcates into the $D1-BD- D2$ critical point, but the $D1-BD-D2$ state is not strictly stable in the sense that the second variation of the rLdG energy of the $D1-BD-D2$ tends to zero as $h\to \infty.$ We believe this to be a generic feature of multi-block critical points for which the prism is effectively partitioned into blocks of $z$-invariant critical points of \eqref{p_energy}, when $\P^t$ and $\P^b$ are $z$-invariant minimisers of \eqref{p_energy}. The multi-block critical point will effectively be constant near the top and bottom surfaces, separated by a thin transition layer near the middle of the prism. We do not observe any other mixed 3D critical points for this particular choice of $\P^t$ and $\P^b$.

In contrast, when we choose $\P^t$ and $\P^b$ to be non energy-minimising $z$-invariant solutions of \eqref{eq:EL} or critical points of \eqref{p_energy} with higher energy, the solution landscape is richer and we obtain multistability or multiple stable rLdG critical points, along with multiple unstable rLdG critical points. Here, it is not energetically preferable to have constant block structures near the top and bottom surfaces.
 In the $\lambda \to \infty$ limit
, we construct approximating profiles for the different admissible configurations by exploiting the non-uniqueness of the boundary conditions in the director framework (captured by the director angle $\theta$). This exercise has a two-fold benefit - these limiting profiles provide good initial conditions for numerical solvers, and in some cases, are good approximations to the numerically computed stable rLdG critical points. 

We generalise some of the analysis for $V_4$ to generic $V_K$, and we largely focus on a numerical computation of rLdG critical points on $V_6$, for different choices of $\P^t$ and $\P^b$.
On $E_K$, we conjecture that the $Para$ states are the $z$-invariant minimisers of \eqref{p_energy} subject to the tangent lateral boundary conditions for $K$ even, with two diagonally opposite splay vertices. For $K$ odd, we conjecture that the $Meta$ states are the $z$-invariant energy minimisers, for which the splay vertices are furthest apart. If $\P^t$ and $\P^b$ belong to the class of $z$-invariant energy minimisers, then we conjecture that $\P^b - BD-\P^t$ is the unique stable mixed 3D critical point on $V_K$ for $K>4$. The $BD$-state is an unstable $z$-invariant critical point of \eqref{p_energy} with two low-order lines (and/or point defects) in the interior of $E_K$. A distinguishing feature of these mixed 3D critical points are defect lines connecting the splay vertices on the top and bottom of $V_K$, running throughout the prism $V_K$. There are multiple possibilities for the multiplicity and locations of these defect lines, and this could be a key driving factor for multistability in 3D. We restrict ourselves to a special temperature $A = -B^2/3C$, largely to facilitate comparisons between our 2D work in \cite{han2020reduced, han2021solution} and the 3D study on polygonal prisms in this manuscript. We speculate that this work can be generalized to arbitrary $A<0$, although there will be technical difficulties and we will need to work with $\Q_c$ in \eqref{eq:Qc} that have a non-constant eigenvalue associated with $\zhat$ i.e. LdG critical points with three degrees of freedom. However, we do expect Proposition~$1$, some of our asymptotic results in the $h\to 0$ and the $\lambda \to \infty$ limits, and the quasi-stable 3D LdG critical points with multi-block structures, to be generic for all $A<0$.

The key question is - can every mixed 3D critical point be related to a 2D pathway between $\P^t$ and $\P^b$ on the 2D solution landscape on $E_K$? The answer is negative. For example, in Fig. \ref{lambda_1}, we only find two solutions, $D1-BD-D2$ and $D1-WORS-D2$ on $V_4$. These solutions are constructed by two pathways $D1\to BD\to D2$, and $D1\to WORS\to D2$ respectively. However, there is another pathway $D-J-R-J-D$, and we cannot find a corresponding 3D mixed critical point on $V_4$. Similarly, on $V_6$, we do not find 3D mixed critical points constructed by the pathways, $P-M1-M-M1-P$ or $P-Ring-P$. Equally, we find some mixed 3D critical points which do not correspond to a pathway on the 2D solution landscape for e.g., we find a 3D solution $TRI-TR-TRI$ in in Fig. \ref{fig:TRI}(c) on $V_6$, for which the middle layer $TR$ is not a rLdG critical point on 2D hexagon.

The overarching question is - what are the hidden compatibility conditions between $\P^t$ and $\P^b$, such that some 2D pathways correspond to mixed 3D critical points and some 2D pathways do not correspond to mixed 3D critical points on $V_K$, with Dirichlet boundary conditions. This is a deep question and requires extensive work, but our work offers good examples and insights which could be foundational for future work on these lines.

\section{Acknowledgments}
This work was supported by the National Key R$\&$D Program of China 2021YFF1200500, the National
Natural Science Foundation of China 12225102, 12050002, 12226316, and the Royal Society Newton
Advanced Fellowship awarded to L. Zhang and A. Majumdar. Y. Han gratefully acknowledges the
support from a Royal Society Newton International Fellowship. A.Majumdar is supported by a
Leverhulme Research Project Grant RPG-2021-401, a Leverhulme International Academic Fellowship
IAF-2019-009, the Humboldt Foundation and a University of Strathclyde New Professors Fund. The authors would also like to thank the Isaac Newton Institute for Mathematical Sciences for support and hospitality during the programme "Uncertainty Quantification and Stochastic Modelling of Materials" when work on this paper was undertaken. This work was supported by EPSRC Grant Number EP/R014604/1.

\appendix
\section{Appendix: Numerical details}
We use the saddle dynamics (SD) method \cite{yin2019high, zhang2022sinum,luo2022sinum}, which has been successfully used to efficiently compute the critical points, to find the reduced Landau-de Gennes critical states in three-dimensional prisms with fixed Dirchlet boundary conditions on the top and bottom surfaces. A critical point $\P$ is an index-$k$ saddle point for which  $\nabla^2 E(\P)$ has exactly $k$ negative eigenvalues: $\lambda_1 \leqslant \cdots \leqslant \lambda_k<0$, corresponding to $k$ unit eigenvectors $\hat{\_v}_1,\cdots,\hat{\_v}_k$ subject to $\big\langle{\hat{\_v}_i}, \hat{{\_v}}_j \big\rangle = \delta_{ij}$, $1\leqslant i, j \leqslant k$.

The SD for finding an index-$k$ saddle point $\P$, (denoted by $k$-SD) is defined as,
\begin{equation}
  \left\{
  \begin{aligned}
  \dot{\P}&=- (\I-2\sum_{i=1}^k {\_v}_i{\_v}_i^\top)\nabla E(\P), \\
    \dot{\_v}_i&=-   (\I-{\_v}_i{\_v}_i^\top-\sum_{j=1}^{i-1}2{\_v}_j{\_v}_j^\top)\nabla^2 E(\P) \_v_i,\ i=1,2,\cdots,k ,\\
  \end{aligned}
  \right.
\label{eq: SD}
\end{equation}
where $\I$ is the identity operator. To avoid evaluating the Hessian of $E(\P)$, we use the dimer
\begin{equation}
  h(\P,\_v_i)=\frac{\nabla E(\P+l\_v_i)-\nabla E(\P-l\_v_i)}{2l}
\end{equation}
as an approximation of $\nabla ^2 E(\P)\_v_i$, with a small dimer length $2l$. By setting the $k$-dimensional subspace $\mathcal{V}=\text{span} \big \{ \hat{\_v}_1,\cdots,\hat{\_v}_k \big \}$, $\hat{\P}$ is a local maximum on $\hat{\P}+\mathcal{V}$ and a local minimum on $\hat{\P}+\mathcal{V}^\perp$, where $\mathcal{V}^\perp$ is the orthogonal complement of $\mathcal{V}$.

The dynamics for $\P$ in \eqref{eq: SD} can be written as
\begin{equation}
\begin{aligned}
\dot{\P}&=\left(\I-\sum_{i=1}^k {\_v}_i{\_v}_i^\top\right) \left(-\nabla E(\P)\right)+ \left(\sum_{i=1}^k {\_v}_i{\_v}_i^\top\right) \nabla E(\P) \\
&= \left( \I-\mathcal{P}_{\mathcal{V}}  \right)\left(-\nabla E(\P)\right)+ \mathcal{P}_{\mathcal{V}} \left(\nabla E(\P)\right),
\end{aligned}
\end{equation}
where $\mathcal{P}_{\mathcal{V}}\nabla E(\P)=\left(\sum_{i=1}^k {\_v}_i{\_v}_i^\top\right)\nabla E(\P)$ is the orthogonal projection of $\nabla E(\P)$ on $\mathcal{V}$. Thus, $\left( \I-\mathcal{P}_{\mathcal{V}}  \right)\left(-\nabla E(\P)\right)$ is a descent direction on $\mathcal{V}^\perp$, and $\mathcal{P}_{\mathcal{V}} \left(\nabla E(\P)\right)$ is an ascent direction on $\mathcal{V}$. 

The dynamics for $\_v_i, i=1,2,\cdots,k$ in \eqref{eq: SD} can be obtained by minimizing the $k$ Rayleigh quotients simultaneously with the gradient type dynamics,
\begin{equation}
\min_{{\_v}_i}\text{  }\left<\_v_i, \nabla ^2 E(\P)\_v_i\right>,\ \text{s.t.}\ \left<\_v_i,\_v_j\right>=\delta_{ij}, \ j=1,2,\cdots,i,
\end{equation}
which generates the subspace $\mathcal{V}$ by computing the eigenvectors corresponding to the smallest $k$ eigenvalues of $\nabla^2 E (\P)$.

In the calculation of critical points in 3D domain, we may encounter a ill-conditioned problem, as the 3D structures like $D-BD-D$ and $D-WORS-D$ have small absolute eigenvalue when $h$ is large enough in Fig. \ref{lambda_1}, which reflects the subtle energy change when the middle slice moves up and down.
Therefore, we use a stable numerical scheme, the semi-implicit scheme for the gradient flow of $\P$ with the Barzilai-Borwein step size \cite{barzilai1988two} for the time discretization. The non-dimensionalized prism domain $V_K$ is discretised into triangular prism or cuboids with mesh size $\delta_x \leqslant 1/32$, using finite difference method for cuboid and hexagonal prism. The finite element method is used to calculate the minimisers in pentagonal prism. We apply a single-step Locally Optimal Block Preconditioned Conjugate Gradient (LOBPCG) method \cite{knyazev1987convergence} to renew the unstable eigendirections instead of the gradient type dynamics in \eqref{eq: SD},
\begin{equation}
  \begin{cases}
    \begin{aligned}
        \frac{P_{1l,n+1}-P_{1l,n}}{\Delta t_n}=& \Delta_{\delta x} P_{1l,n+1}-\lambda^2\left(P_{11,n}^2+P_{12,n}^2-\frac{B^2}{4C^2}\right)P_{1l,n+1}\\
    &+2\sum_{i=1}^k  (D_{\delta_x,P_{11}}E(\P_n)^\top v_{1,n,i}+D_{\delta_x,P_{12}}E(\P_n)^\top v_{2,n,i})v_{l,n,i},\  l=1,2,\\
        \text{Renew } \_v_{l,n,i} & \text{ as } \_v_{l,n+1,i} \text{ with single-step LOBPCG} , \  i=1,2,\cdots,k, \  l=1,2,\\
    \end{aligned}
  \end{cases}
    \label{eq:semi-implicit}
\end{equation}
where $D_{\delta x} E(\P)$ is the discretization of the Frechet derivative. When the point is close to the target critical point enough, i.e., $||D_{\delta x} E(\P)||_F^2\leq 0.01$, we use Newton's method to complete tail convergence with a higher convergence rate  \cite{shi2022hierarchies}. Noting that when the target critical point has small absolute eigenvalue, we use an Inexact-Newton method \cite{dembo1982inexact}, since the ill-conditioned linear equation in Newton iteration is hard to solve exactly. All the symmetric linear equation systems  in \eqref{eq:semi-implicit}, Newton and Inexact-Newton method are solved by The Minimal Residual Method \cite{paige1975solution}.
\bibliographystyle{unsrt}
\bibliography{main}


\end{document}